%% file: SecGamesSP20_eprint.tex
\documentclass[runningheads]{llncs}
\usepackage{graphicx}
\PassOptionsToPackage{hyphens}{url}\usepackage{hyperref} \usepackage{csquotes}
\usepackage{caption} \usepackage{mdwlist} \usepackage{paralist}
\usepackage{tabularx} \usepackage{subcaption} \usepackage{graphicx} 
\usepackage{amsmath,amssymb} \usepackage{algorithm} \usepackage{algpseudocode}
\usepackage{listings} \usepackage{courier} 
\usepackage{multirow} \usepackage{booktabs} 
\usepackage{balance} 
\usepackage{epstopdf}

\algnewcommand\INPUT{\item[\textbf{Input:}]}%
\algnewcommand\OUTPUT{\item[\textbf{Output:}]}%
\DeclareFontFamily{U}{mathx}{\hyphenchar\font45}
\DeclareFontShape{U}{mathx}{m}{n}{<-> mathx10}{}
\DeclareSymbolFont{mathx}{U}{mathx}{m}{n}
\DeclareMathAccent{\widebar}{0}{mathx}{"73}

\usepackage{tikz}
\usetikzlibrary{automata, positioning}
\usetikzlibrary{shapes.geometric,automata,arrows,positioning,calc}
\tikzstyle{startstop} = [rectangle, rounded corners, minimum width=1cm, minimum height=1cm,text centered, draw=black]
\tikzstyle{io} = [trapezium, trapezium left angle=70, trapezium right angle=110, minimum width=2cm, minimum height=1cm, text centered, text width=3cm, draw=black]
\tikzstyle{process} = [rectangle, minimum width=1cm, minimum height=1cm, text centered, draw=black]
\tikzstyle{decision} = [diamond, minimum width=1cm, minimum height=0.5cm, text centered, text width=1.5cm, draw=black]
\tikzstyle{arrow} = [thick,->,>=stealth]
\tikzstyle{connector} = [circle, minimum width=0.5cm, minimum height=0.5cm, text centered, draw=black]

\usepackage[utf8]{inputenc}
\usepackage[T1]{fontenc}
\usepackage{flexisym} 
\usepackage{pifont}

\graphicspath{ {./figures/} }

\usepackage{xargs}                      

\usepackage{todonotes}
\usepackage{xspace}
\def\addauthnote#1#2{%
\expandafter\def\csname#1\endcsname##1{\todo[inline,color=#2]{#1: ##1}\xspace}
\expandafter\def\csname#1m\endcsname##1{\todo[color=#2]{#1: ##1}\xspace}
}
\addauthnote{saeed}{blue!25}
\addauthnote{marten}{orange!25}

\usepackage{listings}
\newcommand{\pluseq}{\mathrel{+}=}

\newcommand*{\mytiny}[1]{\scalebox{0.80}{\ensuremath#1}}

\usepackage{float}
\usepackage{cases}

\usepackage{randomwalk}

\usepackage[noadjust]{cite}

\begin{document}
%
\title{Toward a Theory of Cyber Attacks}
%
\titlerunning{Toward a Theory of Cyber Attacks}
%
\author{Saeed Valizadeh\thanks{This work was funded by NSF grant CNS-1413996 ``MACS: A Modular Approach to Cloud Security''.} \and Marten van Dijk 
}
%
\authorrunning{\hfill}
%
\institute{\href{https://scl.uconn.edu/}{Secure Computation Laboratory},\\
	 Computer Science and Engineering Department, \\
	University of Connecticut, Storrs CT 06269, USA\\
\email{\{saeed.val,marten.van\_dijk\}@uconn.edu}\\
}
\maketitle              
\begin{abstract}
  We provide a general methodology for analyzing defender-attacker based ``games'' in which we model such games as Markov models and introduce a capacity region to analyze how defensive and adversarial strategies impact security. Such a framework allows us to analyze under what kind of conditions we can prove statements (about an attack objective $k$) of the form ``if the attacker has a time budget $T_{bud}$, then the probability that the attacker can reach an attack objective $\geq k$ is at most $poly(T_{bud})negl(k)$''. We are interested in such rigorous cryptographic security guarantees (that describe worst-case guarantees) as these shed light on the requirements of a defender's strategy for preventing more and more the progress of an attack, in terms of  the ``learning rate'' of a defender's strategy.  We explain the 
damage an attacker can achieve by a ``containment parameter'' describing the maximally reached attack objective within a specific time window. 


\end{abstract}
\section{Introduction} \label{sec:intro} 
\input{sections/intro}

\section{Background and Related Work}\label{sec:backg} 
\input{sections/bg}

\section{Game Modeling}\label{sec:gameModel&Analysis} 
\input{sections/game}

\section{Security Analysis}\label{sec:secAna.sim}
\input{sections/SecAnalysis}

\section{Case Studies}\label{CaseStdy}
\input{sections/case_study}

\section{Concluding remarks and future directions} \label{sec:con} 
\input{sections/conclusion}

\newpage
\bibliographystyle{splncs04} 
\bibliography{sections/references}

\appendix
\section{Security Analysis of Case studies}\label{sec:ScenariosAnalysis}
\input{sections/appendix}

\section{Proofs}\label{sec:proofs}
\input{sections/appendix_2}
\end{document}

%% file: sections/intro.tex
Cyber attacks targeting individuals or enterprises have become a predominant part of the computer/information age. Such attacks are becoming more sophisticated (qualitative aspects) and prevalent (quantitative aspects) on a daily basis \cite{cisco:Misc}. The exponential growth of cyber plays and cyber players necessitate the inauguration of new methods and research for better understanding the ``\emph{cyber kill chain}'', particularly with the rise of advanced and novel malware (e.g., Stuxnet \cite{jin2018snapshotter}, WannaCry ransomware crypto worm \cite{cisco:Misc}, the Mirai \cite{bertino2017botnets} and its variants \cite{mirVar:Misc}) and the extraordinary growth in the population of Internet residents, especially connected Internet of Things (IoT) devices.

Mathematical models can help the security research community to better understand the threat and therefore being able to analyze the attacker's conducts during the lifetime of a cyber attack. The sparse amount of research on modeling and evaluating a defensive systems' efficiency (especially from a security perspective), however, warrants the need for constructing a proper \emph{theoretical framework}. Such a framework allows the community to be able to evaluate the defensive technologies' effectiveness from a security standpoint.
In this regard, a proper model is needed to capture the interactions between the two famous players of network security games i.e., a defender (taking advantage of common security tools and technologies such as Intrusion Detection and Prevention Systems (IDPSes), Firewalls, and Honeypots (HPs)) and an attacker (and possibly its agents) who takes actions to reach its attack objective(s) in the game. Modeling only one player by itself (e.g., opportunistic/targeted attacks, or a defensive system such as an IDPS) without acknowledging the other party's capabilities, set of actions, and strategies would be unjustifiable. Hence, a realistic model should take both parties' characteristics, objectives, and actions into consideration, as they are typically oppositional. 
Game theoretic methods have been proposed to model the interactions between an attacker and a defender in a network environment \cite{lye2005game,van2013flipit,chen2009game,alpcan2006intrusion,carroll2011game}. Although insightful, as figuring out the best attack-response strategy is in the focal point of a game-theoretic analysis, such models suffer from a lack of general applicability due to the case-specific assumptions made to reduce the complexity of the problem leading to a case-specific game with a corresponding  set of players' actions and payoffs. 
 Therefore, the developed models become ineffective in the representation of general settings. 

\noindent
{\bf Our Approach.} 
We introduce a Markov Game (MG) framework and methodology for modeling various computer security scenarios prevailing today, including opportunistic, targeted and multi-stage attacks. We are particularly interested in situations in which each players' progress in the game can be viewed and modeled as an incremental process. From the viewpoint of the attacker, a progressive adversarial move associated with a probability distribution can bring it one step closer to its desired winning state (i.e., the attack objective). Similarly, by monitoring the system and capturing the adversarial moves, the defender can also get one step closer to its objective which could be different based on various attack/defense scenarios. For instance, the defender's goal could be generating a detection signal that indicates the presence of a coordinated set of activities that are part of an Advanced Persistent Threat (APT) campaign \cite{milajerdi2018holmes}, developing a signature for a malware attack to halt next adversarial moves \cite{newsome2005polygraph}, or shuffling its resources after enough number of samples (evidence) is obtained as in a Moving Target Defense strategy (MTD) \cite{saeed2016markov}.

\noindent
{\bf Contributions \& Results.} Our contributions are summarized as follows:

\begin{itemize}
	\item We introduce the notion of learning in cybersecurity games and describe a general ``game of consequences'' meaning that each player’s (mainly the attacker) chances of making a progressive move in the game depends on its previous actions. More specifically, as a consequence of the adversarial imperfect moves in the game, the other party, i.e., the defender has the opportunity to incrementally learn more and more about the attack technology. This learning enables the defender to be able to seize and halt following adversarial moves in the game, in other words, containing the attackers’ progress in the game.
	We argue that such learning is possible since the actions that need to be taken by an adversary in each cycle of a cyber attack's life inevitably entails abnormal and suspicious events and activities on both host and network levels.

	\item Unlike game theoretic methods which commonly focus on finding the best attack-response strategies for the players, we, however, with a cryptographic mindset, mainly focus on the most significant and tangible aspects of sophisticated cyber attacks: (1) the amount of time it takes for the adversary to accomplish its mission and (2) the success probabilities of fulfilling the attack objectives.
    Therefore, our goal is to translate attacker-defender interactions into a well-defined game so that we can provide rigorous cryptographic security guarantees for a system given both players' tactics and strategies. We generalize the basic notion of computational security of a cryptographic scheme \cite{lindell2014introduction} to a \emph{system} which must be defended and introduce a \emph{computationally secure system} in which a ``given \emph{system} is $(T_{bud},k,\epsilon)$\emph{-secure} (for $T_{bud},k > 0$, and $\epsilon \leq 1$), if any adversary limited with a (time) budget at most $T_{bud}$ succeeds in reaching the attack objective $k$, with probability at most $\epsilon$''.  We study under which circumstances the defense system can provide an effective response that makes $\epsilon = poly(T_{bud})negl(k)$.

	 \item By modeling the learning rate of the defender as a function $f(l)$ (which is the probability of detecting and halting a next adversarial move) where $l$ represents the number of adversarial moves and actions so far observed and collected, we present the following results:
	\begin{itemize}
		\item[$-$] If $f(l)$ does not converge to $1$, then we call such a learning rate stagnating, and the adversary can reach its winning state (attack objective) $k$ within $poly(k)$ time budget.
		\item[$-$] If $f(l)\rightarrow 1$ and more specifically $1-f(l) = O(l^{-2})$, then the probability of the adversary winning the game is proven to be polynomial in the attacker's time budget and negligible in $k$. 
		\item[$-$] The above results hold for time budget measured in logical Markov model transitions or physical time (seconds), and also hold for learning rates that remain equal to $f(l)=0$ until a certain number $l=L^*$ adversarial moves have been collected after which learning starts.
		\item[$-$] The adversary may not reach its desired winning state of $k$, but will reach  a  much smaller $k_c$ (called containment parameter) of the order $O(L^* + \log T_{bud})$, where $T_{bud}$ is the attacker's time budget, for $1-f(l) \rightarrow 0$.
		
		\item[$-$] We conclude that at a metalevel the adversary needs to find one attack exploit/vector for which the defender cannot find a fast enough increasing learning rate and the defender needs to have a learning system/methodology in place which should be able to reach fast enough learning rates for any possible attack exploit/vector. In practice infrastructures are being compromised and this means that, as given by our general framework, and presented theoretical and numerical analysis, effort is needed in order to understand associated learning rates and why these are $\leq 1-\Omega(l^{-a})$ for some `$a$' too small  or have $L^*=O(k)$ (i.e., learning starts too late).
	\end{itemize}
	
	\item The presented adversarial/defender game in terms of a Markov model is general and fits most practical scenarios as our extensive overview shows. We walk through different case studies to show how the presented modeling can be applied to real-world cyber attacks. 
 
\end{itemize}

Based on our theoretical analysis and its applicability to practical scenarios we understand when an adversary with time budget $T_{bud}$ can be contained and prevented from reaching its attack objective $k$  with probability $\geq 1- poly(T_{bud})negl(k)$. As a consequence, if such containment is in place, then this allows one to trust cryptographic protocols/systems which assume (for proving their security) attackers that have not reached $k$ (and have a key renewal / refresh operation that can be called every $T_{bud}$ seconds in order to push adversaries back to their initial state).

To the best knowledge of the authors, there has been no work so far considering the notion of learning in adversarial games and its impact on players chances of prosperity in the game. Moreover,  through this presented framework, we are able to extend the idea of security capacity introduced in \cite{saeed2016markov} and provide the definition of a \emph{security capacity region} as a metric for gauging a defensive system’s efficiency from a security perspective (by presenting rigorous cryptographic security guarantees in terms of the security capacity region).

\noindent
{\bf Paper Organization.}
The rest of this paper is organized as follows:
Section \ref{sec:backg} briefly reviews  related work, especially modeling cyber attacks as a game between an attacker and a defender.
The interactions of the defender and the adversary are
described as a stochastic network security game in section \ref{sec:gameModel&Analysis}.
The security analysis of the introduced game is presented in section \ref{sec:secAna.sim} followed by the corresponding parametric analysis of the game for various learning functions and system parameters to explore the effects of such variables on attack-defense objectives. In Section \ref{CaseStdy}, we walk through a few attack-defense scenarios to show how our presented framework can be applied to real-world examples. We finally conclude the paper and present future work in section \ref{sec:con}. 


%% file: sections/bg.tex
Modeling cyber attacks as a game between a defender and an attacker is a classical research problem. Here, we review some of the research that has been done in this area which we find the most relevant (current state-of-the-art) and connected to this presented work. 

The intrusion detection problem in heterogeneous networks consisting of nodes with various security assets is studied in \cite{chen2009game}. The expected behaviors of rational attackers in addition to the optimal defender strategy are derived by formulating the attacker-defender interaction as a non-cooperative game. The paper concludes that sufficient resources for monitoring the environment and proper system configuration at the defender side are two necessary conditions of efficiently protecting the network. Similarly, the interaction of players is modeled as a general-sum stochastic game in \cite{lye2005game} in which Lye and Wing studied three different attack-response scenarios including defacing a website, stealing confidential data, and launching a Denial of Service (DoS) attack. The authors though left the richer and more complex scenarios for future works including a more capable defender with a more extensive action set for attack detection and prevention purposes. Such a stronger defender could potentially lure the attacker and learn the attack technology by setting up defensive agents such as honeypots within the environment. Carroll and Grosu \cite{carroll2011game} consider such stronger defense strategies, i.e., taking advantage of camouflage techniques (e.g., disguising a regular system as a honeypot or vice versa) by investigating the effects of deception on players' interaction using a signaling game. 

Motivated by the rise of advanced persistent threats, van Dijk \textit{et al.} \cite{van2013flipit} introduced the {F{\footnotesize{LIP}}}{I{\footnotesize{T}}} game to model the interaction of two players competing in a race of maximizing the amount of time each is in control of a shared computing resource while minimizing their total cost (associated with each player move). Strongly dominant strategies for both players (if there exist such strategies) are determined based on different employed attack strategies. Also, they provided general guidance on how and when to implement a cost-effective defense strategy.

Valizadeh \textit{et al.} \cite{saeed2016markov} introduced the concept of ``security capacity'' as a metric for gauging the effectiveness of an MTD strategy. The interactions of attacker and defender in dynamic environments\footnote{Those with changing system configurations and therefore attack surfaces} is modeled by probabilistic algorithms and characterized by a Markov chain. In particular, they showed how the probability of a successful adversary defeating an MTD strategy is related to the amount of time/cost spent by the adversary. The relationship between the attack success probability and the time it takes to reach the attack objectives (i.e., the winning state for the adversary) is then translated into the security capacity concept:  ``the security capacity of an  MTD game (a defense system) is at least $c$ if the probability that the attacker wins in the first $T = 2^t$ time steps is $\leq 2^{-s}$ for all $t+s = c$''. 
Connell  \textit{et al.} \cite{connell2018performance} used a similar approach as \cite{saeed2016markov} to model the attacker-defender interactions via a Markov chain in dynamic environments. In particular, a quantitative analytic model is proposed for evaluating the performance of MTD schemes, and the availability of resources in the environment and a method is recommended for maximizing a utility function that takes the tradeoffs between security and performance into consideration.

For a review of existing game-theory based solutions for network security problems see \cite{liang2013game,manshaei2013game}. For instance, Liang \textit{et al.} \cite{liang2013game} summarized the presented game models' application scenarios (both cooperative and non-cooperative games) under two categories: attack-defense analysis, and security measurement. Manshaei \textit{et al.} \cite{manshaei2013game} however, surveyed the use of game theory in addressing diverse forms of privacy and security problems in mobile and network applications. The studied works are organized in six main categories: physical and MAC layer security, security of self-organizing networks, intrusion detection systems, anonymity and privacy, network security economics, and cryptography.


%% file: sections/game.tex
In this section, we introduce a general network security game based on the interactions between an adversary, its agents (e.g., bots controlled by a botmaster), and a defender who is equipped with a logically centralized defense system\footnote{For instance, network/host-based intrusion detection and prevention systems, honeypots, etc. One motivation for our work is the advent of Software-Defined Networking (SDN) in which the entire network infrastructure can be controlled from a centralized software controller.} implemented in the network. We explain how this game can be directly mapped to an equivalent Markov model and in the next section, we provide the security analysis for this Markov model as well as presenting compelling interpretations of the role of both players' strategies in their probability of winning the introduced game. 

\subsection{The Game of Consequences}

\noindent
{\bf [Game Setup]}
In almost any sophisticated and persistent cyber attack, the adversary continues the attack until its attack objective is satisfied. However, due to the incomplete and imperfect information of the players (in this case the attacker), and the probabilistic nature of an attack's success rates, the attacker's objective usually cannot be achieved in only a few numbers of moves/attempts. Adding a defender to this picture leads to an unceasing game unless the attacker decides not to play anymore (i.e., dropping out of the game, whether the attack objective is satisfied or the chances of making a progressive move become overwhelmingly small). We model the interactions of players in such scenarios as a stochastic game. 

To create a realistic mathematical model, we make reasonable simplified assumptions from both attack and defense perspectives. This is due to the significant level of freedom in attack design and technology and the complexity and diversity of defense mechanisms and systems.
For instance, when it comes to evaluating a defense system's efficiency (especially an IDPS), the accuracy, performance, completeness, fault tolerance, and timeliness properties should be taken into considerations as the top five criteria \cite{debar2000introduction}. However, for this study, we believe considering all the playing factors in modeling such systems makes the model excessively complicated (and possibly inaccurate). For this reason, we mainly focus on \emph{security-related} concerns and specifications, i.e., the accuracy and completeness (dealing with false alarms and detection rates) of the system. This means the ability of the system in detecting malicious behaviors and parties in the environment, and taking effective actions to foil such incidents, regardless of its architecture, used methods, or its impact on the system performance.
Therefore, assuming that the non-security related traits are ideal (e.g., no latency, high performance, unlimited bandwidth, and fault tolerance), we are dealing with a defense system which is neither entirely accurate nor complete, as every information technology system suffers from security shortcomings and deficiencies.

In this regard, we consider a general system state $(i,l)$ which only captures the security-related parameters. The players start the game at the state $(i=0,l=0)$ in which it represents the inauguration of the attack and the zero-knowledge of the defender regarding the attack technology at time zero. From the adversary's perspective, its view of the system state $i$ will change if only it makes a progressive move in the game. On the other hand, by noticing that any adversarial move (whether fruitful or fruitless) can potentially be observed and captured by the defender (for instance via the defense system's agents and sensors implemented in the environment), we model the defender as an incremental online learning process\footnote{Another motivation for this work is the introduction and widespread embrace of Automated machine learning (AutoML) which provides the opportunity for automatic (and possibly distributed) attack learning, detection, and prevention.}, meaning that
the defender's view $l$ gets updated as a consequence of discovering an adversarial move (i.e., an attack sample). This increasing knowledge of the attack technology enables the defender to correctly detect and halt a new incoming malicious action with some probability $f(l)$ (true positive). It is also possible that the system fails in detecting an adversarial move (or it falsely labels it as benign) with probability $1-f(l)$ (false negative). Note that we do not care about the occurrence of false positives (categorizing a benign activity as malicious or abnormal), as they only play a role in the system performance and the defender's detection cost meaning that such incidents do not change the security state of the system (neither the attacker's nor the defender's view).

\noindent
{\bf [Attacker]} 
The attacker's objective is to reach a winning state of the form $(k, *)$ within a limited time budget $T_{bud}$. We emphasize that $k$ or in other words the attack objective differs from scenario to scenario. For instance, a malware propagator or a botmaster desires to push $k$ (in this case, the total number of infected nodes) as high as possible\footnote{Or at least 5\% of the total number of vulnerable hosts, as \cite{provos2004virtual} shows via simulation that in order to have hopes that no more than 50\% of the total vulnerable hosts ever become infected, patching process must begin before 5\% of such population become ever infected.}. On the contrary, in an advanced persistent threat, the attacker might prefer to stop playing after reaching a much smaller $k$ to minimize the attack detection probability since once the target machine has been identified, the attacker should use as few infections as possible to decrease the chance of exposing the operation. An adversary trying to access distributed information on a set of nodes within the network terminates the attack as soon as its mission is accomplished. In a similar fashion, to construct a hitlist, the attacker will conclude its reconnaissance after a compiled list of vulnerable nodes is constructed.



To reach the winning state, at any time step, the attacker issues a move associated with a success probability $p$ (purely depends on attack strategy and not the defender's maneuvering) that can potentially lead to incremental progress in the game, i.e.,  getting one step closer to the final attack objective  $(i = i +1)$. This enables us to represent the adversarial move in a single transition in the Markov model which will be explained in section \ref{ss.markovModel}. For generality, we consider a state-dependent success probability, i.e., $p_i$, since, intuitively, there could be cases in which the chances of making a progressive move depends on the current state of the attacker in the game. For instance, consider an attacker who is trying to locate vulnerable nodes within an address space of size $N$, while $i$ out of $K$ vulnerable hosts are already found, this means that the probability of hitting a new vulnerable node is $p_i = (K-i)/(N-i)$. Also, we assume the adversary is aware of its current state in the game $i$ in $(i,l)$ for some $l \geq 0$ but it does not know the defender's knowledge of the attack expressed as $l$.



Moreover, we assume that at each time step, the interaction of the attacker (or an attacker agent) with the system which is encapsulated as an adversarial move, can potentially leak some information to the defender (if observed) regarding the attack technology and methodology. This is because regardless of the mastery and skillfulness of an attacker, its taken actions during any phase of the attack, inevitably lead to abnormal and suspicious incidents and events on both host and network levels. Hence, almost all security tools and technologies rely on the existence of such attack signal indicators for detection and prevention purposes. For instance, unusual port usage \cite{inoue2009automated}, irregular system call sequences \cite{kolosnjaji2016deep}, and the occurrence of pointer value corruption on a host’s process memory \cite{liang2005fast} are amongst a few events that can happen as a result of adversarial moves on host levels.  On the network levels, such conducts include but are not limited to an increase in the network latency, high bandwidth usage, suspicious traffic on exotic ports, irregular scan activity, simultaneous identical domain name system (DNS) requests, and Internet relay chat (IRC) traffic generated by specific ports. 

However, a skillful attacker leveraging different anti-malware and IDPS evasion techniques\footnote{Obfuscation methods, fragmentation and session splicing, application/protocol violations, and DDoS attacks \cite{marpaung2012survey} are among the most common techniques used by the adversaries to decrease the attack detection probability or to increase the chances of attack prosperity.} is capable of minimizing the detection possibilities and therefore reducing the chances of leaking attack information at each adversarial move by sneaking through the defensive system. This skillfulness can provide the adversary more time to play the game as each attack sample has a lower probability of being discovered and in extreme cases, it can ideally be different for each adversarial move making the defender's job surely difficult.
For instance, an attacker taking advantage of a polymorphic or metamorphic malware (in contrast with a monomorphic one) will disclose minimum information by encrypting the content and obfuscating the instruction sequences for each connection respectively leading to more propagation time available to the attacker. 
Therefore, dealing with a skillful adversary, if the defense system captures an attack sample, this property intuitively leads to maximizing the time and effort to push out a solution to stop the attack and generate an attack signature. More importantly, it means that only a few samples are not enough for attack signature generation. Note that this property helps to minimize the information leakage for the suspicious traffic flow detection, signature generation, malware analysis, and reverse engineering processes.

\begin{figure}
	\begin{center}
		\scalebox{0.60}
		{\includegraphics{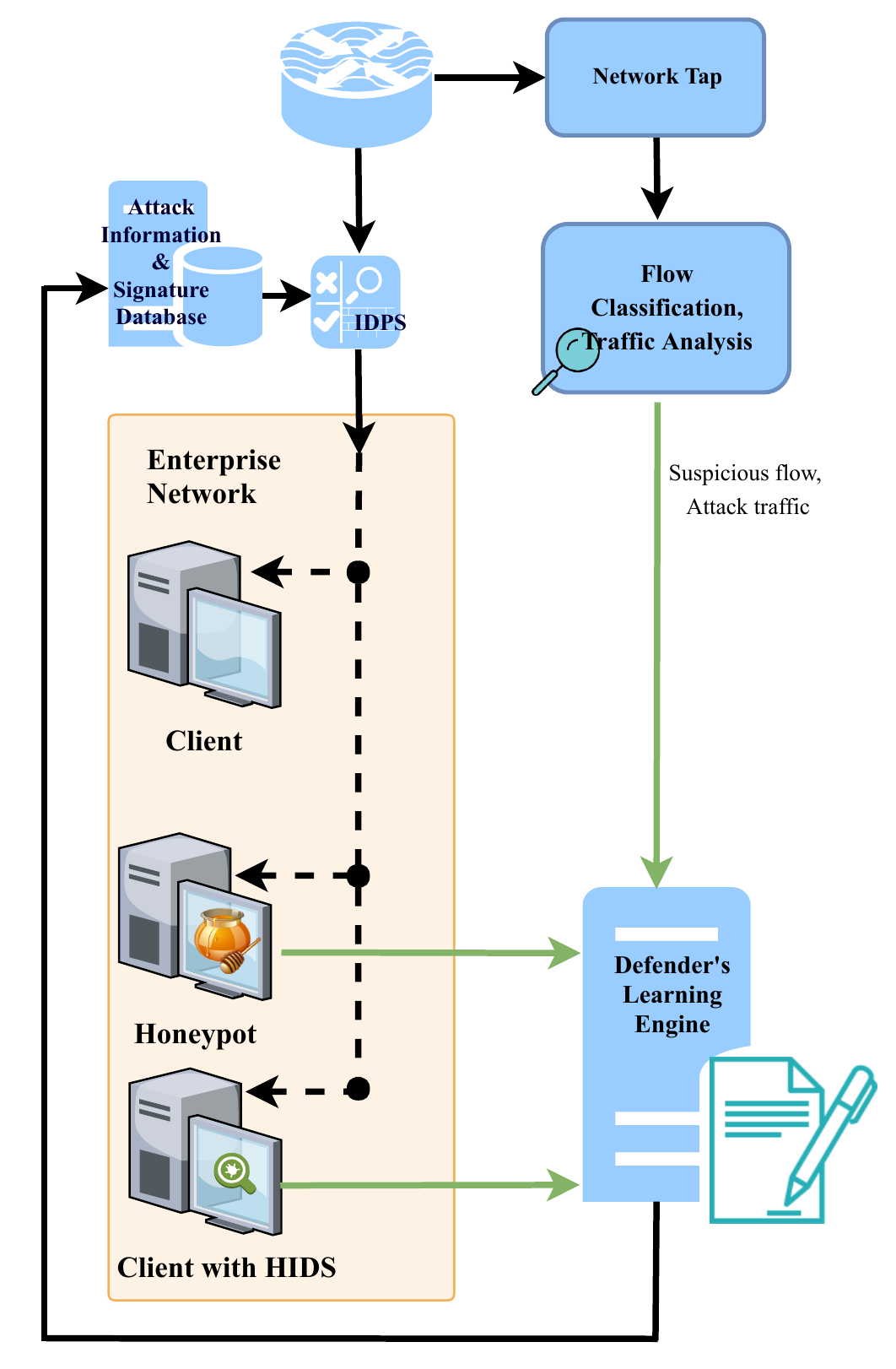}}
		\caption{Defense system architecture}
		\label{fig.ids_scenario}
	\end{center}
\end{figure}

\noindent
{\bf [Defender]} 
The defender's objective is to learn more and more regarding the attack technology and methodologies to be able to cease a next incoming adversarial move. In this regard, it monitors the environment for attack detection and prevention purposes via common security tools and technologies. Intrusion detection (and possibly prevention) systems are amongst the most common type of defense technologies used for monitoring, incident identification, attack detection, and prevention on both host and network levels\footnote{The first one is known as host-based intrusion detection system (HIDS), and the latter is known as network intrusion detection system (NIDS).}. Virtual and physical sensors are a common component of such systems used for data collection and analysis. Also, the defender can take advantage of electronic decoys known as honeypots for attack information assembly.

We give the defender the opportunity to learn from observed adversarial moves and actions.
The idea is that in the early phases of the game, the defender's knowledge of the attack is limited (almost non-existent). As the game proceeds, the defender's detection rates will be enhanced by witnessing enough number of attack samples. In order to study the learning rate of the defender or the time that the system is finally trained and is able to detect the adversary's actions with probability almost $1$, we consider a cumulative time-varying function $0\leq f(.) \leq 1$ as the detection rate, which is a function of the total number of times an adversary agent's activity is captured or in other words, an attack sample is given to the defender. The function $f$ indicates the probability that given $l$ samples so far, the system detects and halts a new incoming adversarial move. The idea is the more samples are given to the defender, the more accurate would be its attack signatures which immediately reflects in the detection rate which is used by the system to filter future adversarial moves. 

In this regard, we assume that all incoming traffic passes through an inline NIDS implementation\footnote{Meaning that network traffic directly passes through the IDPS sensors, and the system is capable of session snipping, dropping/rejecting suspicious network activities and altering and sanitizing the malicious content within a packet.} which gives the defender the ability to be able to block an adversarial move with some probability $f(l)$ in which $l$ represents the current realization of the attack by the defender (i.e., based on so far observed and collected adversarial moves and attack information). Moreover, we assume that a copy of all incoming traffic is given to the defender’s traffic analysis and classification engine via a network tap or a spanning port which provides the opportunity of detecting suspicious flows with some probability. For simplicity, however, we model this offline analysis as a probabilistic sampling process \cite{sperotto2010overview} in which each incoming adversarial move can be correctly sampled/labeled as suspicious with probability $\gamma$ (therefore $l = l +1$), and with $1-\gamma$, the defender misses the adversarial move or it falsely labels it as benign. Note that $\gamma$ reflects the classifier's accuracy in which it is the probability that an incoming packet will be marked as suspicious traffic conditioned on the fact that it is indeed malicious.
If the defender fails in bringing an adversarial move to a halt and therefore the attacker proceeds within the system with probability $1-f(l)$, the defender still has the opportunity to learn (i.e., $l = l +1$) regarding the attack technology with probability $h$ via the defense agents (e.g., HPs, HIDS) implemented in the environment. Hence, if the adversary deals with an electronic decoy, or as a consequence of attack activities, on an endpoint device equipped with a HIDS, the defender can learn regarding the attack meaning that $l = l+1$.

In summary, a progressive move by the defender ($l = l+1$) leads to an increase in its future detection rates $f$. This function gets updated in two different manners: $(1)$ if the attack traffic is correctly labeled as suspicious on the network level with probability $\gamma$ or $(2)$ if the attack activity leaks information on host levels with probability $h$ (whether via a honeypot or a host-based intrusion detection system installed on a fraction of defender systems\footnote{Notice that, in case of a honeypot, any adversarial move will lead to an increase in the defender's knowledge of the attack $l$ since all incoming communications with an HP is suspicious, while for a HIDS, not all the adversarial moves might be observed by the defense system. For simplicity, we use a single parameter $h$ as the probability of gaining information on host levels (mainly HPs). However, one can easily separate the learning from HPs and the HIDS agents by considering another parameter for HIDS agents $h'$ or in general, a cumulative function $f_{h'}$  for the host-level attack information and signatures (observed and gained from HIDS agents).}). This incremental learning of the attack enables the defender to be able to bring a new incoming adversarial move to a halt with probability $f(l+1)$.
Notice that the adversary's chance of making a progressive move in the future time steps decreases as the $f$ function's value increases over time.

\subsection{Markov Model of the  Game}\label{ss.markovModel}

The above description of the adversary-defender interactions can immediately be translated into a Markov model.  In summary, at each time step, the state of the system can be identified with a tuple $(i,l)$ in which $i$ represents the adversary's progress in the game, and $l$ delineates the defender's knowledge of the attack methodologies and technology. The adversary (or one of its agents) makes a move associated with a success probability $p_i$. Note that $p_i$ is indeed the attack success probability while there exists no opponent player in the environment (i.e., the defender). In the meantime, the defender can block an adversarial move, if it can detect it correctly or it can learn the attack technology if the attack is observed via the defense system or the attacker is in touch with a defense system agent (e.g., a honeypot) implemented in the environment.

\begin{figure}
	\centering
	\begin{tikzpicture}
	\tikzset{
		>=stealth',
		node distance=3.2cm and 3.2cm ,on grid,
		every text node part/.style={align=center},
		state/.style={minimum width=2.00cm,
			draw,
			circle},
	}
	\node[state]             (c) {$i,l$};
	\node[state, right=of c] (r) {$i+1,l$};
	\node[state, above=of c] (b) {$i,l+1$};
	\node[state, above=of r] (rr) {$i+1,l+1$};
	\draw[every loop]
	(c) edge[auto=left]  node {$m_{2}(i,l)$} (r)
	(c) edge[ auto=right] node {$m_3(i,l)$} (rr)
	(c) edge[auto=left] node {$m_4(i,l)$} (b)
	(c) edge[loop left]             node {$m_1(i,l)$} (c);
	\end{tikzpicture}
	\caption{Possible transition probabilities at each state $(i,l)$ in the Markov model of the game}
	\label{markovModel}
\end{figure}
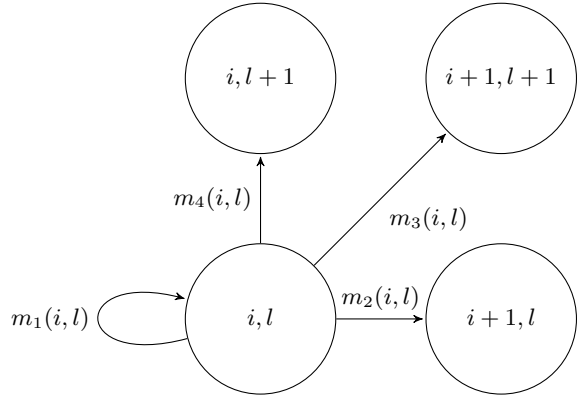
%
Let us denote the occurrence of the event that an adversarial move comes to a halt by the defense system with $\emph{F}$ and the occurrence of the event that it gets sampled (or labeled as suspicious) by the defense system with $\emph{S}$. Notice that these two events are independent, as we are assuming that the straining and filtering happens online for each detected incident whereas sampling deals with a copy of the actual traffic for attack analysis and classification.
Therefore the following outcomes are possible at each timestep:
\begin{itemize}
	\item $(F\wedge \widebar{S})$: meaning that the adversasry's move comes to a standstill by the defender (based on current knowledge of the attack technology) but not sampled which happens with probability $f(l)(1-\gamma)$, as a result, the Markov chain stays at state $(i,l)$ 
	\item  $(\widebar{F}\wedge \widebar{S})$: The adversasry's move gets neither filtered nor sampled, there exist three possible scenarios for this case:
	\begin{itemize}
		\item[$-$] if the adversary is dealing with an electronic decoy (i.e., a  honeypot), the defender has the oppurtunity to learn from this interaction meaning that the Markov model transits from state $(i,l)$ to $(i,l+1)$ with probability $(1-f(l))(1-\gamma)h$
		\item[$-$]  if the attacker makes a progressive move with probability $p_i$ then the Markov chain transits to state $(i+1,l)$ with probability $(1-f(l))(1-\gamma)p_i$.  Note that this is a ``perfect move'' for the adversary as in which the attacker gets one step closer to the attack objective while the defensive system was not able to detect and locate the adversary's activities or it falsely labeled the activity as benign (false negative) meaning that no information and knowledge regarding the attack technology is leaked to the defender
		\item[$-$] and if the attacker is not dealing with a defensive agent and not also make a progressive move with probability $(1-p_i-h)$ the Markov chain stays at the current state with probability $(1-f(l))(1-\gamma)(1-p_i-h)$
	\end{itemize}
	\item $(F\wedge S)$: If the attacker's move gets both filtered and sampled, the chain transits to the state $(i,l+1)$ with probability $\gamma f(l)$
	\item $(\widebar{F}\wedge S)$:  If the attacker's action does not get filtered but it gets sampled
	\begin{itemize}
		\item[$-$] for a successful move with probability $p_i$, the Markov chain transits to the state $(i+1,l+1)$ with probability $(1-f(l))\gamma p_i$
		\item[$-$] and for an unsuccessful move with probability $(1-p_i)$, the Markov chain transits to the state $(i,l+1)$ with probability $(1-f(l))\gamma (1-p_i)$
	\end{itemize}
\end{itemize}
In summary, the transition probabilities can be expressed by:
\begin{eqnarray}\label{eq.transitions}
m_1(i,l) &=& Prob(\mytiny{(}i,l\mytiny{)}\rightarrow\mytiny{(}i,l\mytiny{)})=f(l)(1-\gamma)+\nonumber\\ & &
(1-f(l))(1-\gamma)(1-p_i-h)\\ 
m_2(i,l) &=& Prob(\mytiny{(}i,l\mytiny{)}\rightarrow\mytiny{(}i+1,l\mytiny{)}) = \nonumber\\ & & (1-f(l))(1-\gamma)p_i \\
m_3(i,l) &=& Prob(\mytiny{(}i,l\mytiny{)}\rightarrow\mytiny{(}i+1,l+1\mytiny{)}) = \nonumber\\ & & (1-f(l))\gamma p_i\\
m_4(i,l)&=& Prob(\mytiny{(}i,l\mytiny{)}\rightarrow\mytiny{(}i,l+1\mytiny{)}) = \nonumber\\ & & (1-f(l))(1-\gamma)h+  f(l)\gamma +\nonumber\\ & &  (1-f(l))\gamma (1-p_i)
\end{eqnarray}

Fig. \ref{markovModel} depicts the possible transitions at each state of the Markov model. The initial state of the game is $(0,0)$ which it shows the starting point of the attacker in the game, and the zero knowledge of the defender at the beginning (e.g., a zero-day vulnerability/exploit). As the game evolves, reaching state $k$ is in favor of the adversary which can happen via horizontal $m_2$ or diagonal $m_3$ transitions. In the meantime, the defender's knowledge of the attack technology and signatures improves through transitions via $m_3$ and $m_4$ in which the number of attack samples provided to the defender, i.e., $l$ increases and the function $f$ gets updated consequently.

%% file: sections/SecAnalysis.tex
\subsection{Adversarial Containment} \label{AC-AC}

The goal of the adversary is to reach a winning state $(k,l)$ for some $l\geq 0$ within a limited time budget $T_{bud}$ in which $k$ represents the attack objective which is defined based on various attack scenarios\footnote{For instance, $k$ could be the total number of nodes the adversary wishes to compromise before using the collective of these nodes in the next attack phases (e.g., botnet applications).}, and time budget $T_{bud}$ expresses the number of transitions in the Markov model, i.e., the total number of adversarial moves.  

The objective of the defender is to learn more and more regarding the attack technologies (e.g., malicious attack payloads),  in order to be able to recognize and halt future moves by the attacker. That is, the defender hopes to push level $l$ -- the number of so far discovered adversarial moves-- higher and higher, leading to a higher probability $f(l)$ of successfully detecting and filtering a next adversarial action. The ultimate goal of the defender is to have a good enough learning mechanism in place such that the probability of reaching $k$ is negligible in $k$ (i.e., exponentially small in $k$) and only polynomial in $T_{bud}$. In other words, the adversary cannot reach a winning state in practice.

We define $w(k,T_{bud})$ as the probability of the adversary to reach a state $(k,l)$ for some $l\geq 0$ within time budget $T_{bud}$. In the next subsections we analyze this probability. We assume a parameter $p=p_i$ for all $i$, i.e., we simply give the adversary the advantage of using $p=\sup_i p_i$ instead of the smaller individual $p_i$ values. 
We want to find a tight upper bound and determine for which model parameters $\gamma$, $h$, $p$, and learning rate $f(l)$ the upper bound proves the desired asymptotic 
\begin{equation} w(k,T_{bud})=poly(T_{bud}) negl(k). \label{eq:as} \end{equation} 
 We notice that in this case, even though the adversary cannot reach a large $k$ because $w(k,T_{bud})$ is negligible in $k$, the adversary may reach a (much) smaller $k$. We are interested in the order of magnitude of this value as a sub-objective the adversary is able to reach. For this reason, we define $k_c$ as the solution of
$$ w(k_c, T_{bud}) =1/2.$$
For instance,  $k_c$ could reflect the order of magnitude of the number of compromised nodes to which the adversary is {\em contained}. By using $w(k,T_{bud})=poly(T_{bud}) negl(k)$ we can show\footnote{This implies that with significant probability the adversary can reach an attack objective at most $O(\log T_{bud})$ (e.g., compromising at most $O(\log T_{bud})$ nodes). Therefore, if the adversary would not have the advantage of using $p=\sup_i p_i$, then it would (compromise less nodes and) only use parameters $p_0, \ldots, p_{O(\log T_{bud})}$ in the Markov game. If these first few $p_i$ values do not differ much, then a constant $p$ becomes a realistic assumption which does not give the adversary that much advantage.}   `{\em containment parameter}' $k_c= O(\log T_{bud})$. 

Our analysis will show that for a stagnating learning rate (i.e., $1-f(l)$ remains larger than some constant $\tau>0$) the adversary cannot be contained ($w(k,T_{bud})$ does not follow the desired asymptotic). For a learning rate $f(l)\geq 1-O(1/l^2)$ which approaches $1$ fast enough, we are able to prove the desired asymptotic.
We also develop tight lower and upper bounds on $w(k,T_{bud})$ and show an efficient algorithm which is able to compute these bounds with $O(k\cdot T_{bud})$ complexity.


\subsection{Time budget measured in real time}\label{Real_Time}

 We first discuss the effect of translating the time budget $T_{bud}$, measured in Markov model transitions, to a time budget which is measured in real time (seconds) as this makes sense in practice. If the adversary wins as soon as $k$ nodes are compromised and if each Markov model transition takes $\Delta$ seconds, then at most $k$ moves can be made by each of the compromised nodes {\em in parallel} in $\Delta$ seconds. With a time budget of $T$ seconds, this implies $T_{bud}\leq T k/\Delta$. This means that the desired asymptotic (\ref{eq:as}) translates into
$$ w(k,T_{bud})\leq poly(T k/ \Delta) negl(k).$$
This is $poly(T)negl(k)$ and again implies that in practice the adversary cannot reach $k$ since the probability of reaching $k$ is negligible in $k$ and only polynomial in the physical time budget $T$. 

In some scenarios each adversarial move costs a certain amount, say $C$ USD, for example due to having to  rent a virtual machine in order to remotely launch or execute an attack move. In such a case the time budget $T_{bud}$, measured in Markov model transitions, directly translates into $C\cdot T_{bud}$, the amount of USD the adversary is restricted to. Such conversion allows one to reason about economically constrained adversaries.

\subsection{Stagnating learning rate} \label{stagnating}

Suppose that the learning rate stagnates in that $1-f(l)\geq \tau$ for some $\tau>0$. This means that $f(l)\leq 1-\tau$ and as a result at least a fraction $\tau$ goes undetected. 

\begin{theorem}[Stagnating learning rate]\label{Theo.stag}
	The attacker's probability $w(k,T_{bud})$ of winning   a game in which the defender's learning rate stagnates in that $1-f(l)\geq \tau$ is given by 
	\begin{equation}
w(k,T_{bud})
\geq 1-k(1-\tau p)^{T_{bud}/k}.
	\end{equation}
\end{theorem}

	For a stagnating learning rate, the attacker almost always wins the security game, i.e., $w(k, T_{bud})\approx 1$ even for relatively small time budgets $T_{bud}=c\cdot k$ for some multiplicative factor $c\geq 1$. This is because $1-w(k,T_{bud}) \leq k(1-\tau p)^{T_{bud}/k}$ is exponentially small in $T_{bud}/k=c$. 

In order to have the desired asymptotics (\ref{eq:as}), the learning rate should not stagnate, i.e., we need $1-f(l)\rightarrow 0$ for $l\rightarrow \infty$.

\subsection{Formula for $w(k,T_{bud})$}

The theorem below characterizes $w(k,T_{bud})$ in a closed form expression. In the next subsection, we show that neglecting the sum over $T$ of $B$ values leads to tight upper and lower bounds. The remaining sum over $A$ values turns out to correspond to a simpler Markov model which can be used to compute this sum in $O(k\cdot T_{bud})$ time using dynamic programming.

\begin{theorem}[Closed Form Winning Probability]\label{closedw}
The probability $w(k, T_{bud})$ of winning is equal to
$$\sum_{L=0}^{T_{bud}-1}\sum_{B=0}^{\min\{L,(k-1)\}}
  \frac{m_2(L)+m_3(L)}{1-m_1(L)} \cdot A 
 \cdot \sum_{T=0}^{(T_{bud}-1)-[(k-1)+L-B]} B $$
 with $A$ and $B$ substituted by
 \begin{eqnarray*}
 A &=&
 \sum_{\substack{b_1,\dots,b_L\in\{0,1\} s.t.\\ \label{p-a}
					\sum_{l=1}^{L}b_l=B}}
\prod_{l=0}^{L-1} \bigg( b_{l+1}\frac{m_3(l)}{1-m_1(l)} + (1-b_{l+1})\frac{m_4(l)}{1-m_1(l)} \bigg) \\
&& \cdot \sum_{\substack{g_0\geq 1,\dots,g_L\geq 1 s.t.\\ \label{p-b}
					\sum_{l=0}^{L}(g_l-1)=(k-1)-B}}
  \prod_{l=0}^{L} \bigg( \frac{m_2(l)}{1-m_1(l)}\bigg)^{g_l-1} \mbox{, and } \\
  B &=&	\sum_{\substack{t_0,\dots,t_L s.t.\\ \label{p-c}
					\sum_{l=0}^{L}t_l=T}}
	\prod_{l=0}^{L}			{t_l+g_l-1 \choose t_l}
 m_1(l)^{t_l}(1-m_1(l))^{g_l}.
  \end{eqnarray*}
\end{theorem}

\subsection{Upper and Lower Bounds on $w(k, T_{bud})$} \label{AC-dp}

In the notation of Theorem \ref{closedw}, we define $\bar{w}$ as 
	\begin{eqnarray}
	\bar{w}(k,T_{bud}) 
	&=&
	\sum_{L=0}^{T_{bud}-1}\sum_{B=0}^{\min\{L,(k-1)\}} \frac{m_2(L)+m_3(L)}{1-m_1(L)}\cdot A. 
	\label{UB}
	\end{eqnarray}

\begin{theorem}[Upper and Lower Bounds on $w$]\label{Theo.wbar}
	 For $v\geq k$, we have
	\begin{equation}
	 \bar{w}(k,T_{bud}/v)\cdot (1- (k + \frac{T_{bud}}{v}) (1-\gamma)^{\frac{v}{k}}) \leq w(k,T_{bud}) \leq \bar{w}(k,T_{bud}). \label{LB}
	\end{equation}
\end{theorem}



In particular for $v= k \ln ((T_{bud}+k)/\epsilon) / \ln (1/(1-\gamma))$, we have $w(k,T_{bud}) \geq \bar{w}(k,T_{bud}/v)\cdot (1-\epsilon)$. 
Upper bound  (\ref{UB}) proves that if $\bar{w}(k,T_{bud})$ has the desired complexity, i.e., $\bar{w}(k,T_{bud})=poly(T_{bud}) negl(k)$, then $w(k,T_{bud})$ inherits this desired asymptotics. And vice versa, the lower bound in (\ref{LB}) for the special choice of $v$ proves that if $w(k,T_{bud})$ has desired asymptotics $w(k,T_{bud})=poly(T_{bud}) negl(k)$, then also $\bar{w}(k,T_{bud}) \leq w(k, T_{bud}\cdot v)/(1-\epsilon) = poly(T_{bud} \cdot v)\cdot negl(k) = poly(T_{bud})negl(k)$. We conclude that in order to find out which learning rates lead to the desired asymptotics  for the defender, we only need to study which learning rates lead to $\bar{w}(k,T_{bud})=poly(T_{bud})negl(k)$.

The advantage of studying $\bar{w}(k,T_{bud})$ is having eliminated expression $B$ in Theorem \ref{closedw} which represents the effect of self-loops (and turns out to be very inefficient to evaluate as an exact expression). Probability $\bar{w}(k,T_{bud})$ corresponds to a Markov model without any self-loops and with horizontal transition probabilities $m_2(l)/(1-m_1(l))=(1-\gamma)\alpha(l)$, diagonal transition probabilities $m_3(l)/(1-m_1(l))=\gamma \alpha(l)$, and vertical transition probabilities $m_4(l)/(1-m_1(l))=1-\alpha(l)$, where 

\begin{equation} \alpha(l) = \frac{p(1-f(l))}{\gamma + (1-\gamma)(p+h)(1-f(l))}. \label{eq-alpha} \end{equation}

 Algorithm \ref{alg:wbar}  depicts how $\bar{w}(k,T_{bud})$ can be computed in $O(k\cdot T_{bud})$ time using dynamic programming -- and this is what we use in section \ref{AC-Param} to evaluate $\bar{w}(k,T_{bud})$ for different learning rates.
 

\begin{algorithm}
	\caption{$\bar{w}$ Computation}
	\label{alg:wbar}
	\begin{algorithmic}[1]
		\Function{$\bar{w}$}{$k,T_{bud}$}
		\State{$pr \gets zeros(k-1,T_{bud}-1)$}
		\State{$pr[0,0]\gets 1$}
		\For{$i\gets 0,T_{bud}-1$}
		\If{$i != 0$}
		\State{$pr[0,i]\gets(1-\alpha(i-1))pr[0,i-1]$}
		\EndIf
		\For{$j \gets 1,k-1$}
		\If{$i==0$}\State{$pr[j,i]\gets [(1-\gamma)\alpha(0)]^{j}$}
		\Else \State{$pr[j,i]\gets (1 - \gamma)\alpha(i)pr[j-1,i] + \gamma \alpha(i - 1) pr[j-1,i-1] + (
			1 - \alpha(i - 1))pr[j,i-1]$}
		\EndIf
		\EndFor
		\EndFor
		\State{$\bar{w}\gets 0$}
		\For{$i\gets 0,T_{bud}-1$}
		\State{$\bar{w} \gets \bar{w}+\alpha(i)pr[k-1,i]$}
		\EndFor
		\State \textbf{return}{$~\bar{w}$}
		\EndFunction
	\end{algorithmic}
\end{algorithm}

\subsection{Analyzing Learning Rate $f(l)=1-\frac{d}{(l+2)^2}$}

We further upper bound $\bar{w}(k, T_{bud})$ in the next theorem which shows 
 that small enough learning rates $f(l)\geq 1- d/(l+2)^2$ will attain our desired asymptotic (as explained as a consequence after the theorem).
In Section \ref{AC-phase}  we will discuss what happens if there is an initial period during which the adversary already starts making progress in the game (e.g., compromising nodes) without the defender being aware or not having learned anything. For example, $f(l)=0$ for $l\leq L^*$ and $f(l)\geq 1-\frac{d}{(l-L^*)^2}$ for $l>L^*$. Here,  $L^*$ represents this initial period; only after a sufficient level $L^*$ is reached, i.e., a sufficient number of attack samples (e.g., malware payloads) have been collected, the learning rate increases.

\begin{theorem}[Upper Bound on $\bar{w}$]\label{Theo.wbarUB}
Assume that there exists a $d\geq 0$ such that 
$1-f(l)\leq d$ for all $l\geq 0$, or 
equivalently, $\beta(l)=\sqrt{\frac{1-f(l)}{d}}\leq 1$ for all $l\geq 0$. 
Let $\theta\geq (1-\gamma)/\gamma$. Then, 
		\begin{eqnarray*}
		\bar{w}(k,T_{bud})
		&\leq&  	(1+\theta^{-1}) \sum_{L=0}^{T_{bud}-1} [\theta p d]^{k} \prod_{l=0}^{L} \left(1+ (1+\theta^{-1})\frac{\beta(l)}{1-\beta(l)}\right).
		\end{eqnarray*}

Let 
$$ c =\int_{l=0}^{T_{bud}-1} \frac{\beta(l)}{1-\beta(l)} dl +  \frac{\beta(0)}{1-\beta(0)}.$$
Then, minimizing the above upper bound with respect to $\theta$ proves
$$ c\leq \frac{1-\gamma}{\gamma} k \ \  \Rightarrow  \  \  \bar{w}(k,T_{bud})\leq  (1-\gamma)^{-1}  T_{bud} e^{(1-\gamma)^{-1} c} [\frac{1-\gamma}{\gamma}pd]^{k} .$$
Applying this statement for 
$$ 1-f(l) \leq \frac{d}{(l+2)^2}$$
gives
\begin{eqnarray*}
	&& T_{bud} \leq e^{-1+k (1-\gamma)/\gamma }  \\
	&& \hspace{1cm} \Rightarrow  \  \ \bar{w}(k,T_{bud}) \leq 
\frac{e^{1/(1-\gamma)}}{1-\gamma} \cdot T_{bud}^{1+1/(1-\gamma)} \cdot  [\frac{1-\gamma}{\gamma}pd]^{k} .
\end{eqnarray*}
\end{theorem}

For $f(l)\geq 1-\frac{d}{(l+2)^2}$ with $d<\frac{\gamma}{p(1-\gamma)}$, the theorem shows that if $T_{bud}< exp(k)$, then $w(k,T_{bud})\leq \bar{w}(k,T_{bud})= poly(T_{bud})negl(k)$ satisfying the desired asymptotics. 
 For analyzing other learning rates, the more  general upper bound of the theorem can be used.

\subsection{Delayed learning} \label{AC-phase}

In practice, learning only starts after seeing a sufficient number of attack samples. That is, $f(l)=0$ for $l< L^*$ for some $L^*$ after which $f(l)$ starts to converge to $1$. What would be the attacker state in the game during this initial phase, where no learning takes place (for instance, how many nodes will be compromised)? 

\begin{theorem}[Delayed learning]\label{Theo.delay} Suppose $f(l)=0$ for $l< L^*$.
	Define $u(k^*,L^*)$ as the probability  that $L^*$ is reached by a state $(i,L^*)$ with $i\leq k^*$, where we give the adversary the advantage of an unlimited time budget. Then, for $k^*\geq L^*+1$,
	$$ u(k^*,L^*) \geq 1- L^* [\frac{(1-\gamma)p}{1-(1-\gamma)(1-(p+h))}]^{\frac{k^*-1}{L^*}}. $$

Let $w'(k,T_{bud})$  correspond to learning rate $f(l+L^*)$ as a function of $l$ (e.g., $1-f(l+L^*)\leq \frac{d}{(l+2)^2}$). Then, for $k^*\geq L^*+1$,
\begin{equation} w(k,T_{bud})\leq (1- u(k^*,L^*)) + w'(k-k^*,T_{bud}),\label{eqdelayed} \end{equation}
where $1-u(k^*, L^*)$ is the probability that  $>k^*$ is achieved during the initial phase and where $ w'(k-k^*,T_{bud})$ is the probability that another $k-k^*$ progressive adversarial moves happen after the initial phase with time budget $T_{bud}$. 
\end{theorem}

The lower bound on $u(k^*,L^*)$ can be further simplified by using $p\leq 1-h$ which implies 
$$  u(k^*,L^*) \geq 1- L^* [(1-\gamma)(1-h)]^{\frac{k^*-1}{L^*}}. $$




Since $1-u(k^*,L^*)$ is exponentially small in $k^*/L^*$, this shows that very likely $k^* = O(L^*)$. This makes common sense because initially, every adversarial move has a significant probability of success since the defender's knowledge of the attack is very limited in the beginning (the attack is still unknown to the defender, in other words, no valid attack signature is formed yet). 
If we achieve desired asymptotic (\ref{eq:as}) for $w'(k,T_{bud})$, then the resulting containment parameter for $w(k,T_{bud})$ (see Section \ref{AC-AC}) is $k_c=O(L^* + \log T_{bud})$. It is very important for the defender to keep $L^*$ small. A larger $L^*$ leads to a longer initial phase during which the adversary keeps on making progressive moves of the order of $O(L^*)$.  

\subsection{Capacity Region}

In this subsection, we show that as a result of our analysis we are able to prove statements like the one in the following definition which defines a `capacity region' of combinations of parameters defining the adversarial time budget, the probability of winning, and $k$ which characterizes what it means to be in a winning state. The definition describes what was previously called `desired asymptotics' since a {\em non-empty} capacity region shows that an exponentially small probability of winning $O(2^{-s})$  is achieved if $k=O(s + \log T_{bud})$ implying that the probability of winning is $poly(T_{bud})negl(k)$.


\begin{definition} Suppose there exist vectors $\boldsymbol{\delta}$, $\boldsymbol{\mu}$, $\boldsymbol{\xi}$, and the all-one vector ${\bf 1}$, such that, for all $k\geq 0$ (characterizing a winning state for the adversary) and for all $s\geq 0$ and $t\geq 0$ satisfying
$$ s\boldsymbol{\delta} + t \boldsymbol{\mu} \leq k {\bf 1} + \boldsymbol{\xi},$$
 we have the following security guarantee: If the attacker has a time budget $\leq 2^t$, then his probability of reaching a winning state is at most $\leq 2^{-s}$. 
 
 We say that $(\boldsymbol{\delta}, \boldsymbol{\mu},  \boldsymbol{\xi})$ describes a capacity region for $\{(s,t) \ : \ s,t\geq 0\}$. Its definition implies that for a fixed $k$ and time budget $2^t$, the probability of winning is at most $\leq 2^{-s(k,t)}$, where
 \begin{equation} s(k,t)= \max\{ s\geq 0 \ : \  s\boldsymbol{\delta} + t \boldsymbol{\mu} \leq k {\bf 1} + \boldsymbol{\xi} \}. \label{eqcapreg} \end{equation}
 Similarly, for a  time budget $2^t$, the probability of winning is $\leq 2^{-s}$ if  a winning state is characterized by some $k\geq k(s,t)$, where
 \begin{equation} k(s,t)= \min\{ k\geq 0 \ : \  s\boldsymbol{\delta} + t \boldsymbol{\mu} \leq k {\bf 1} + \boldsymbol{\xi} \}. \label{eqcapreg2} \end{equation}
 \end{definition}
 
This definition generalizes the definition of security capacity used in MTD games \cite{saeed2016markov} where only a single equation of the form $s+t\leq \delta \cdot k$ needs to be satisfied and where $\delta$ is defined as the ``security capacity''. The definition of capacity region may be of interest for other more general defender-adversarial games where a Markov model is used to characterize defender and adversarial moves.\footnote{A more general definition which captures  winning states  characterized by parameter vectors ${\bf k}$, rather than a single $k$, can also be formulated.}


The definition of capacity region is of interest for a couple of reasons. First, it clearly describes the limitation of the adversary: 
The capacity region explicitly characterizes the probability of winning as a function of $t$ and $k$, see (\ref{eqcapreg}). In practice, the attacker's time budget is $\leq 2^t$, limited to some ``small'' $t$, e.g., $t=80, 128$, and using such a fixed $t$ makes the probability of winning only a function of $k$. 

Second, the capacity region can also be used to compute the containment parameter $k_c$, defined as the solution  $k(s,t)$ in (\ref{eqcapreg2}) where $s$ is set to $s=1$ which corresponds to a probability of winning equal to $1/2$. This allows us to understand how 
$$k_c=k(1,t)=O(t)=O(\ln T_{bud})$$ depends on the parameters that describe the overall Markov model.

Third, the definition allows us to trivially compose capacity regions of Markov models corresponding to different learning rates. For instance, in practice, an attacker may use several exploits or attack vectors and use each attack vector to advance its footprint and increase the number of compromised nodes. For each attack vector $i$, the defender develops a learning rate $f_i(l)$. In essence, the defender plays individual games with each attack vector, and the adversary wins the overall game if the individual games lead to $k_i$ compromised nodes such that $k=\sum_i k_i$. The following theorem makes this argument precise.

\begin{theorem} Let $f_i(l)$, $1\leq i\leq a$, be learning rates corresponding to Markov models describing different defender-adversary games that are played simultaneously. Let  $k_i(s,t)$ for the $i$-th game  be defined as in (\ref{eqcapreg2}). Let $T_{bud}\leq 2^t$ be the overall time budget of the adversary. Then the probability of reaching a state with $k_i(s+\ln a,t)$ compromised nodes in the $i$-th game is at most $\leq 2^{-(s+\ln a)}=2^{-s}/a$. Since the transition probabilities in each of the Markov models are independent from the number of compromised nodes, the probability of reaching $k=\sum_{i=1}^a k_i(s+\ln a,t)$ is at most
$$\leq 1- \prod_{i=1}^a (1-2^{-s}/a)\leq 2^{-s}.$$
\end{theorem}

Notice that for $a$ simultaneous adversarial games an exponentially small probability of winning $O(2^{-s})$  is achieved for $k=O(s + \ln a + \ln T_{bud})$. In this parallel game, if the time budget is measured in real time $T$, then (as before) $T_{bud}\leq Tk/\Delta$ for  $k=\sum_{i=1}^a k_i(s+\ln a,t)$ and we require $T\leq \frac{\Delta}{k} 2^{t}$. As a final note,  $k=\sum_{i=1}^a k_i(s+\ln a,t)$ is dominated by  $\max_{i=1}^a k_i(s+\ln a,t) \geq k/a$. In other words, the attacker keeps on searching for an attack vector which can be used to achieve a large fraction\footnote{Due to the defender either starting learning after a too long initial period or, once learning starts,  not learning fast enough (in order to attain the desired asymptotic).} of the desired $k$, and the defender tries to learn how to recognize attack vectors as fast as possible in order to contain these sufficiently.



As an example of  a capacity region, we translate Theorem \ref{Theo.wbarUB} for $1-f(l) \leq \frac{d}{(l+2)^2}$ in terms of a capacity region: 

\begin{corollary} \label{crcor}
For learning rate  $1-f(l) \leq \frac{d}{(l+2)^2}$ we have a capacity region $(\boldsymbol{\delta}, \boldsymbol{\mu},  \boldsymbol{\xi})$ defined by 2-dimensional vectors
\begin{eqnarray*}
\boldsymbol{\delta} &=& (0, q \ln 2), \mbox{ where } q = \left[ \ln(\frac{\gamma}{pd(1-\gamma)})\right]^{-1}, \\
\boldsymbol{\mu} &=& (\frac{\gamma \ln 2}{1-\gamma}, q (\ln 2)(1+\frac{1}{1-\gamma})), \\
\boldsymbol{\xi} &=& (-\frac{\gamma}{1-\gamma}, q(-\frac{1}{1-\gamma} + \ln \frac{1}{1-\gamma}).
\end{eqnarray*}
\end{corollary}

As a second example of working with capacity regions, we translate (\ref{eqdelayed}) in Theorem \ref{Theo.delay} for delayed learning:
 
\begin{corollary} \label{lemdel} Suppose $f(l)=0$ for $l<L^*$ and let $(\boldsymbol{\delta}', \boldsymbol{\mu}',  \boldsymbol{\xi}')$ be the capacity region corresponding to learning rate $f(l+L^*)$ as a function of $l$. Suppose that
$$z= \left[ \ln(\frac{1-(1-\gamma)(1-(p+h))}{(1-\gamma)p})\right]^{-1}\geq 0.$$
Then, first,
 for all $s\geq \max\{0, 1/{z\ln 2} - {\ln (2L^*)}{\ln 2} \}$ and $t\geq 0$ satisfying
$$ s[\boldsymbol{\delta}'+  L^* (\ln 2) z {\bf 1}]
 + t \boldsymbol{\mu}' \leq k {\bf 1} + [\boldsymbol{\xi}' -\boldsymbol{\delta}' + ( L^* (\ln 2L^*) z -1){\bf 1}] $$
and, second, for all\footnote{Notice that  $1/{z\ln 2} - {\ln (2L^*)}{\ln 2}\geq 0$ if and only if $L^*\leq \frac{1}{2}e^{1/z}= \frac{1-(1-\gamma)(1-(p+h))}{2(1-\gamma)p}$, which is large  for small $(1-\gamma)p$.} $0\leq s\leq  1/{z\ln 2} - {\ln (2L^*)}{\ln 2}$ and $t\geq 0$ satisfying
$$ s\boldsymbol{\delta}' + t \boldsymbol{\mu}' \leq k {\bf 1} + [\boldsymbol{\xi}' -\boldsymbol{\delta}' - (L^*+1){\bf 1}],$$
 we have the following security guarantee: If the attacker has a time budget $\leq 2^t$, then his probability of reaching a winning state is at most $\leq 2^{-s}$. 
\end{corollary}

The corollary confirms that the containment parameter for delayed learning (for both cases) is 
$k_c = O(L^* + \log T_{bud})$.


\subsection{Parametric analysis of Learning Rate $f(l)=1-\frac{d}{(l+2)^a}$}\label{AC-Param}

In this subsection, we use Algorithm \ref{alg:wbar} to evaluate $\bar{w}(k,T_{bud})$ based on different learning rates and $T_{bud}=10,000$ logical timesteps. Fig. \ref{fig.wbar} depicts $\bar{w}(k,T_{bud})$ as a function of $k$ for several $\gamma$ and  several learning rates of the form $\alpha(l)= (l+1)^{-a}$. As it can be seen from this plot, a proper defense mechanism (with sufficiently fast converging learning rates and large enough $\gamma$ values) substantially decreases the adversary's chances of making progress in the game. 
 
 Fig. \ref{fig.kc} depicts how the containment parameter $k_c$ (computed using $\bar{w}(k,T_{bud})$) depends on $\alpha(l)= (l+1)^{-a}$. For $a\geq 0.9$, $k_c\leq 20$ is very small making it impossible for the adversary to reach high-value attack objectives (e.g., total number of infected nodes in a malware propagation game).
 
 \begin{figure*}
 	\begin{center}
 		\scalebox{0.65}
 		{\includegraphics{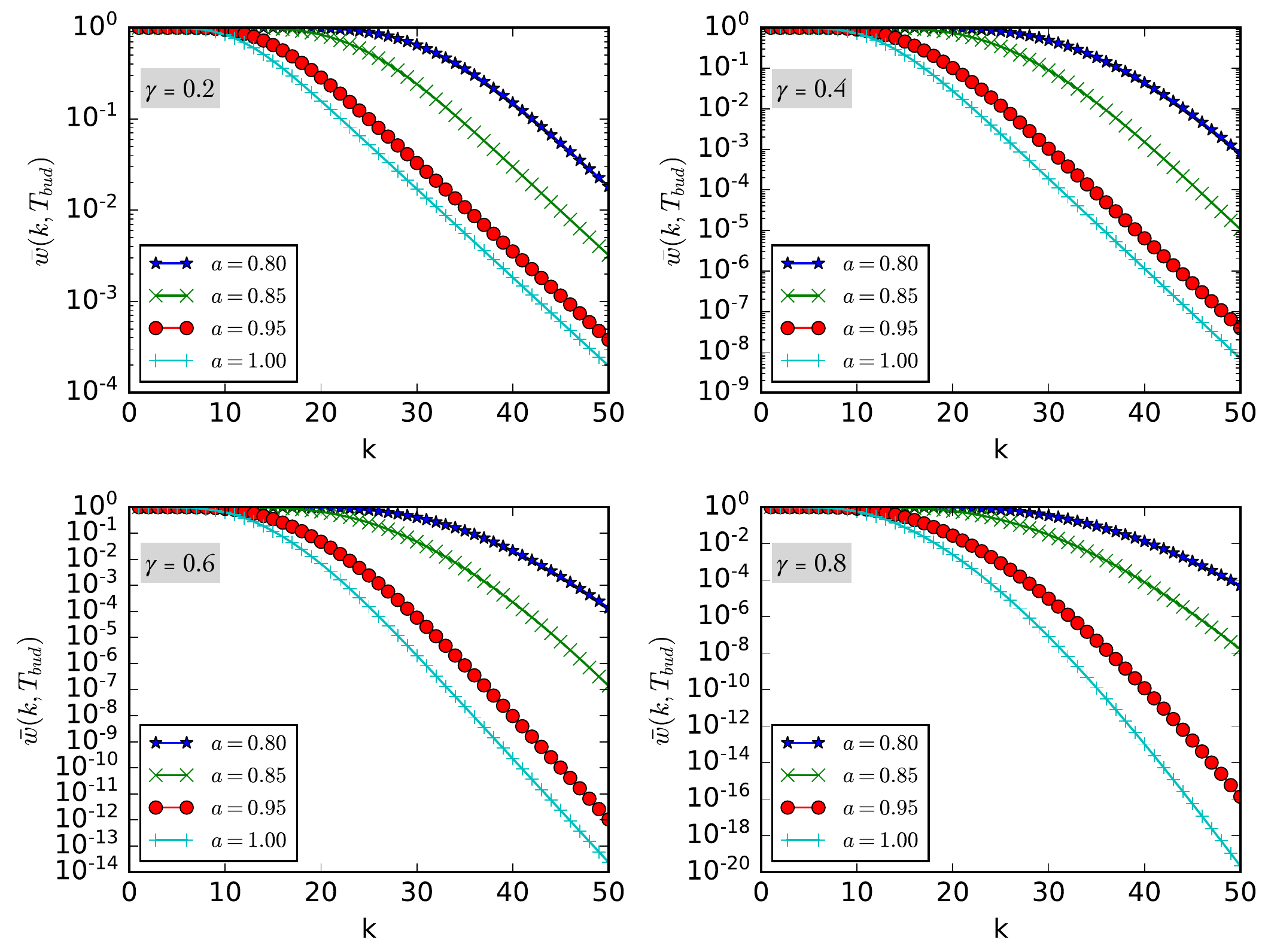}}
 		\caption{$\bar{w}(k,T_{bud})$ based on different values of $\gamma$ and learning functions $\alpha(l) = (l+1)^{-a}$ where $T_{bud} = 10,000$ logical timesteps}
 		\label{fig.wbar}
 	\end{center}
 \end{figure*}
 
 \begin{figure}
 	\begin{center}
 		\scalebox{0.43}
 		{\includegraphics{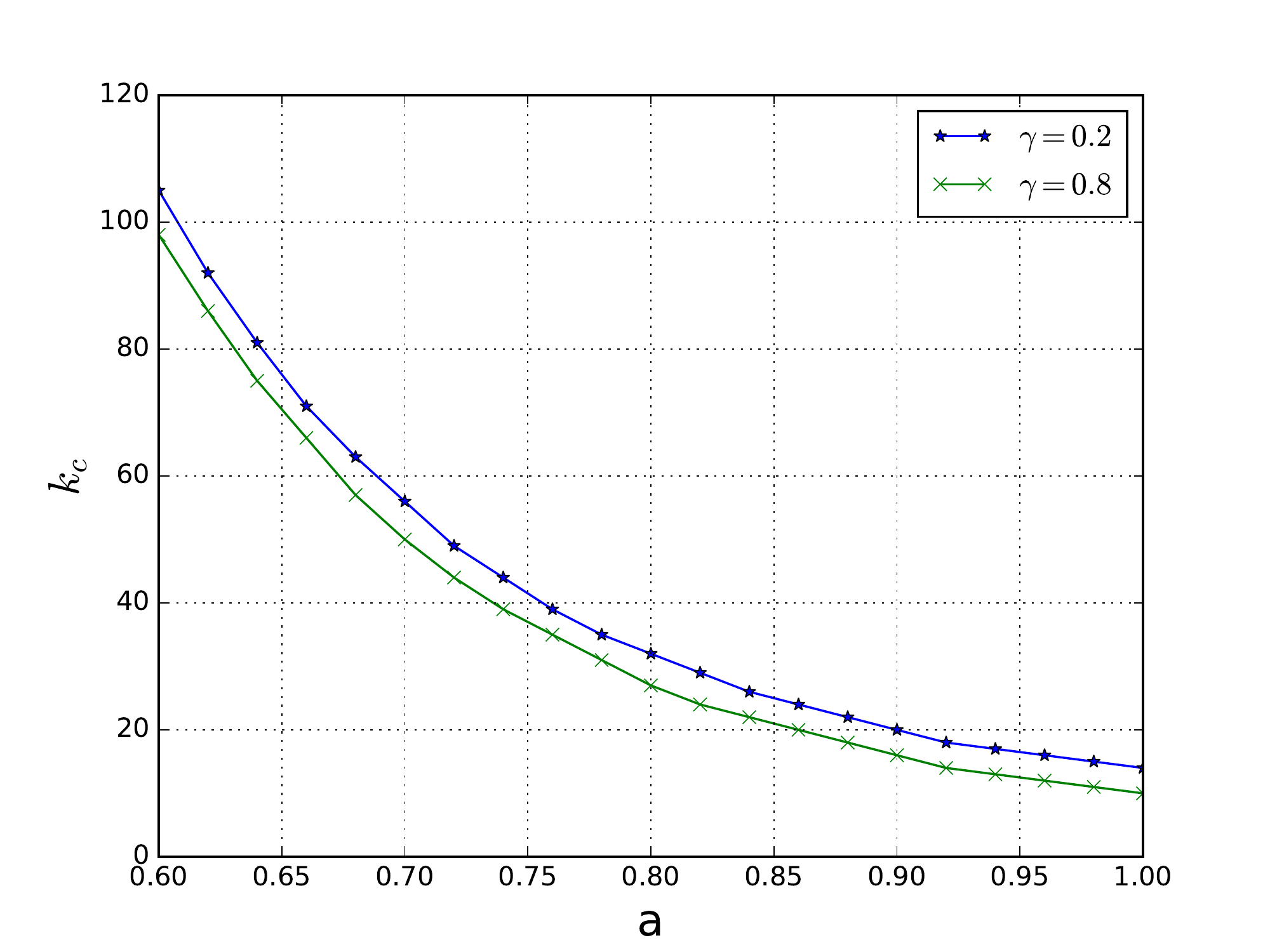}}
 		\caption{$k_c$ for different classes of learning functions $\alpha(l) = (l+1)^{-a}$ where $T_{bud} = 10,000$ logical timesteps}
 		\label{fig.kc}
 	\end{center}
 \end{figure}
 
 Fig. \ref{fig.roleOfp} depicts the chances of attack prosperity in two different cases $p=0.1$, and $p =1.0$ (see (\ref{eq-alpha})) while $\gamma = 0.5$ meaning that half of the adversarial moves will be observed by the defender at the network levels (i.e., will be marked as suspicious on the flow classifier). 
\begin{figure}
 	\begin{center}
 		\scalebox{0.43}
 		{\includegraphics{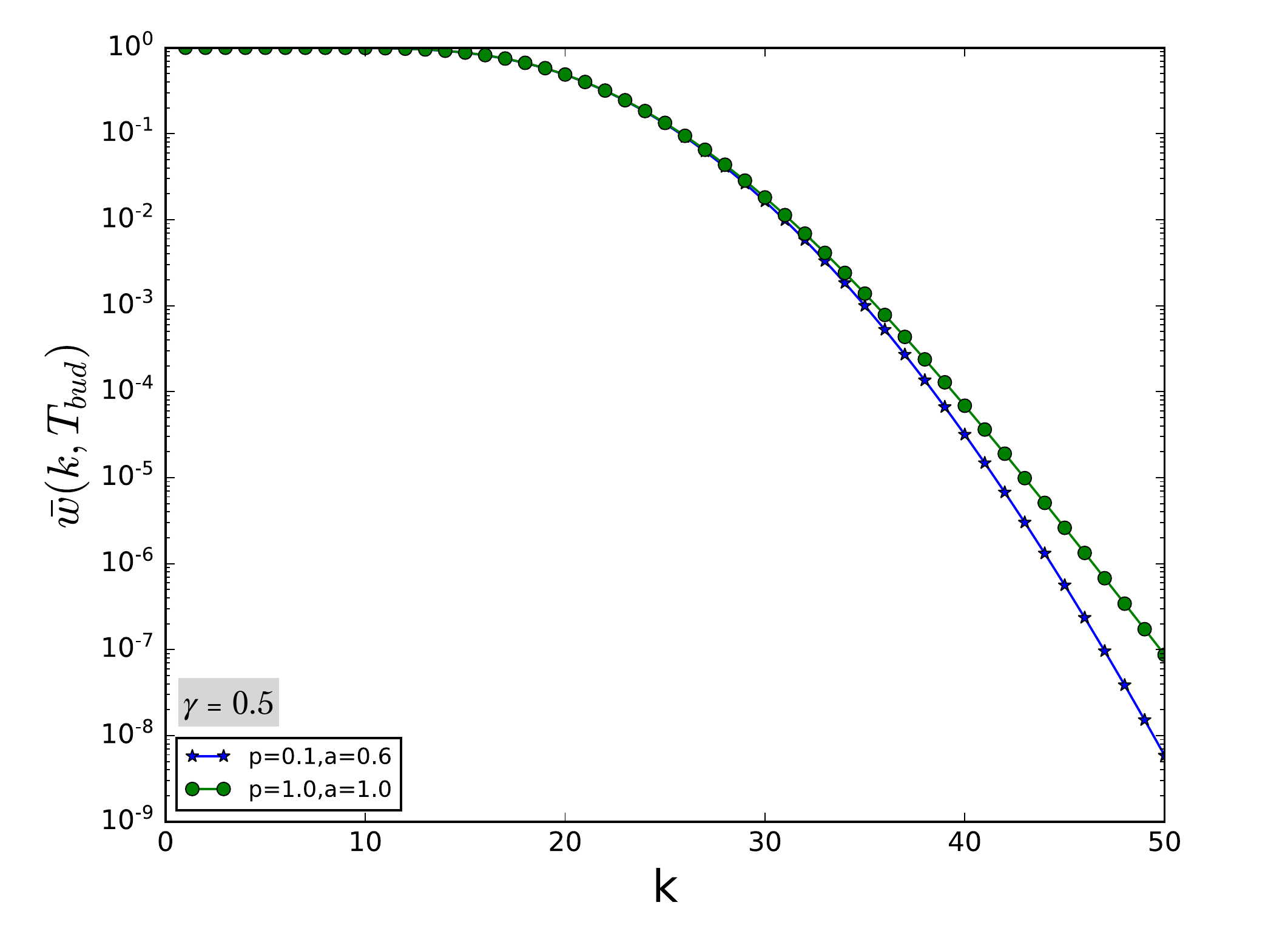}}
 		\caption{$\bar{w}(k,10000)$ for two cases $p=0.1$, and $p=1.0$ with $\gamma=0.5$ and learning function $1-f(l) = (l+1)^{-a}$ where $a=0.6$ and $a=1$ respectively.}
 		\label{fig.roleOfp}
 	\end{center}
 \end{figure}
 
 The impact of delayed learning on the adversarial containment parameter $k_c$  is shown in Fig. \ref{fig.delayed} for two different cases. In the first scenario, the learning starts immediately (with no delay, i.e.,  $L^{*}=0$), and in the second case, the learning commences after missing the first 100 attack samples (i.e., $L^*=100$). 
 
 \begin{figure}
 	\begin{center}
 		\scalebox{0.43}
 		{\includegraphics{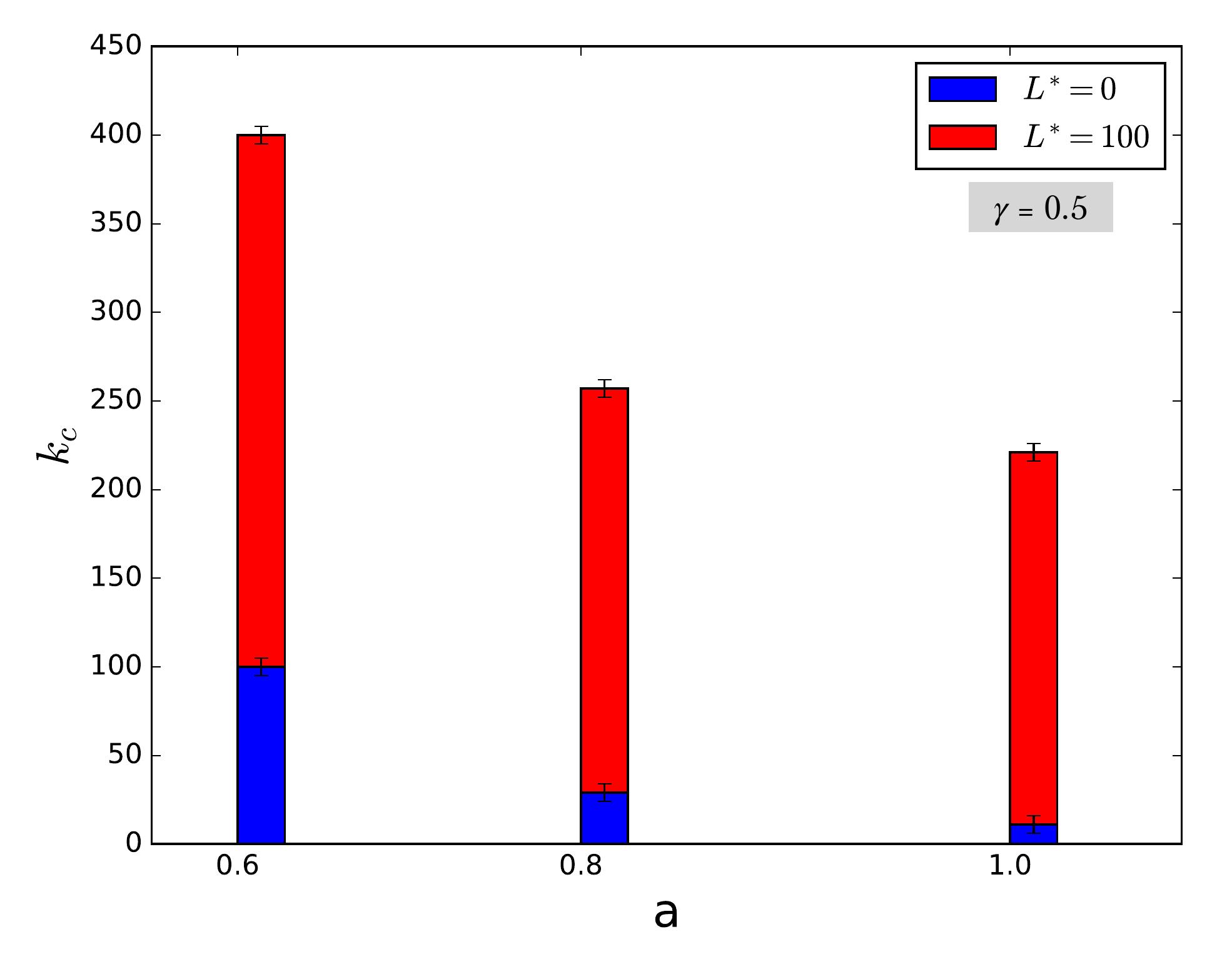}}
 		\caption{The impact of delayed learning on $k_c$ for different classes of learning functions $\alpha(l) = (l+1)^{-a}$ where $T_{bud} = 10,000$ logical timesteps}
 		\label{fig.delayed}
 	\end{center}
 \end{figure}

%% file: sections/case_study.tex
In this section, we walk through two examples of attack-defense games as case studies to show how real-world cyber attack scenarios can be translated into our presented framework. Their security analsyis is worked out in the appendix.

\subsection{Malware propagation and Botnet construction} \label{CS1}

 Consider a statically addressed network of size $\mathcal{N}$ in which a fraction of the hosts $\mathcal{K}$ suffers from a zero-day vulnerability only known to an adversary. The attacker's objective is to locate such hosts in the network, based on a target discovery strategy $\sigma$ (e.g., a random scan scheme, or a hitlist) associated with a probability $p_i$, and infect them during the Window of Vulnerability (WoV) time as the number of vulnerable systems is not yet shrunken to insignificance and the attacker’s exploit is useful in this period. The game starts with one infected machine as the ``Patient Zero''. The attacker desires to take control of at least $k$ out of $\mathcal{K}$ vulnerable hosts, meaning that the winning state for the adversary is defined as $(i,l)=(k,*)$. 
In the meantime, the defender’s objective is to generate an attack signature to be able to filter next adversarial attack traffic (see Fig. \ref{fig.ids_scenario}). The defender’s knowledge of the attack can be increased in two manners: (1) if an adversary agent hits a honeypot with probability $h$ (assuming that there exist in total $\mathcal{H}$ honeypots in the environment, therefore, $h =\mathcal{H}/\mathcal{N}$ for a memoryless blind scan strategy), or (2) if the attack traffic is correctly labeled as suspicious on the network (classifier) level with probability $\gamma$ \footnote{This is a typical architecture in automatic signature generation schemes (e.g., see \cite{kim2004autograph,newsome2005polygraph,li2006hamsa,tang2011signature}).}. In either of these cases, the defender’s detection rate $f(l)$ in which $l$ is the number of so far collected attack samples will be enhanced. This function represents the probability that given $l$ samples so far, the system detects a new incoming malicious packet and filters it on the fly. We also consider a skillful attacker (for instance taking advantage of a polymorphic worm) as defined in previous sections, and therefore, having only a few attack samples is not enough for signature generation.

\begin{algorithm}
	\caption{Malware propagation game simulator}
	\label{alg:sim}
	\begin{algorithmic}[1]
		\Function{Simulator}{$\mathcal{N,K,H},k,\sigma,\varepsilon,\gamma,f(.),\vartheta$}
		\State {$\mathcal{I} = \emptyset$; $l = 0$;}
		\State {droppingOut = False};
		\State {winningState = False};
		\While {not droppingOut and not winningState}
	\State {targetHost $\xleftarrow[\text{}]{p_i} probe(\mathcal{N},\sigma)$};
		\State {filtered,sampled $\xleftarrow[\text{}]{f(l),\gamma}$$transmit$($\varepsilon$,targetHost)};
		\If {(not filtered and targetHost $\in \mathcal{H}$) or sampled} \label{alg.def_start}
		\State{$l\pluseq 1$};
		\State {update $f(l)$};
		\If {$f(l) \geq 1-\vartheta$}
		\State {Output \textit{"Game is finished. Attacker dropped out! Defender won!"}};
		\State {droppingOut = True};
		\EndIf 
		\EndIf \label{alg.def_end}
		\If {targetHost $\in \mathcal{K}$ and not filtered}
		\State $\mathcal{I} = \mathcal{I}~ \cup$ targetHost;
		\If {$|\mathcal{I}| \geq k$}
		\State {Output \textit{"Game is finished. Adversary won!"}};
		\State {winningState = True};
		\EndIf 
		\EndIf
		\EndWhile
		\EndFunction
	\end{algorithmic}
\end{algorithm}

The system state $(i,l)$ represents the total number of infected machines and the total number of captured attack traffic (i.e., malware samples) so far by the defense system. The attacker's view of the system state will change if an agent's effort in infecting a new node is successful (or if it possibly loses its control over an agent).
The defender's view gets updated as a consequence of discovering an adversarial move (i.e., an attack/malware sample). 
Algorithm \ref{alg:sim} shows the above description of the game in which the game simulator takes $\mathcal{N,K,H},k$, and an adversarial target discovery strategy $\sigma$ and an exploit $\varepsilon$, in addition to the detection probability $f(l)$, the sampling rate $\gamma$, and an acceptable ``threshold'' $1-\vartheta$ as the input and outputs if the attacker wins the game. The termination rule is whether the attacker compromises its desired number of hosts $k$ or if the attacker decides not to play anymore\footnote{For instance, the attacker concludes that not enough gain can be made in a reasonable time because of not being able to make a progressive move in the game due to defender's high-value detection rates that keep on improving.}, that is $f(l) \geq 1-\vartheta$. Note that the transmit method returns a tuple, i.e., if the transmitted packet by the attacker got filtered by IDPS or got sampled at the flow classifier level.

\subsection{A Moving Target Defense Game}  \label{CS2}
As another example, we consider the ``Multiple-Target Hiding'' (MTH) game introduced in \cite{saeed2016markov} with minor modifications to the game. In the MTH game studied in \cite{saeed2016markov}, the adversary is interacting with a probabilistic defender taking advantage of a moving target defense strategy in which it reallocates/shuffles its resources at each time step of the game with some probability $\lambda$ (for instance, consider an IP hopping strategy). In order to win the game, the adversary needs to locate $k$ out of $\mathcal{K}$ sensitive targets/resources distributed in the environment while there exist in total $\mathcal{N}$ ``locations''. The defender moves by reallocating a target causing the attacker to redo its search for the locations. The attacker moves by selecting one of the $\mathcal{N}$ possible locations and examining whether it corresponds to one of the targets or not. Notice that MTD strategies are usually costly for the defender as they have an immediate impact on the availability of the resources and system performance. Therefore, instead of a randomized defense strategy (i.e., a probabilistic move at each time step), we consider that the defender issues a move if the number of observed adversarial actions reaches a threshold $L^*$ based on which an attack detection signal is being generated. 

\begin{algorithm}
	\caption{MTD  game simulator}
	\label{alg:MTDsim}
	\begin{algorithmic}[1]
		\Function{Simulator}{$\mathcal{N,K,H},k,\sigma,\gamma,L^*$}
		\State {$\mathcal{I} = \emptyset$; $l = 0$;}
		\State {reallocateResources = False};
		\State {winningState = False};
		\While {not reallocateResources and not winningState}
		\State {observed,targetHost $\xleftarrow[\text{}]{(\gamma,h),p_i} locate(\mathcal{N},\sigma)$};
		\If {observed} \label{alg:MTDsim.def_observe}
		\State{$l\pluseq 1$};
		\If {$l \geq L^*$}
		\State {reallocateResources = True};
		\State {$\mathcal{I} = \emptyset$};
		\State {Output \textit{"Game is finished. Defender won!"}};
		\EndIf 
		\EndIf \label{alg:MTDsim.def_end}
		\If {targetHost $\in \mathcal{K}$}
		\State $\mathcal{I} = \mathcal{I}~ \cup$ targetHost;
		\If {$|\mathcal{I}| \geq k$}
		\State {winningState = True};
		\State {Output \textit{"Game is finished. Adversary won!"}};
		\EndIf 
		\EndIf
		\EndWhile
		\EndFunction
	\end{algorithmic}
\end{algorithm}

The system state $(i,l)$ represents the total number of located target machines by the attacker and the total number of recognized adversarial moves so far by the defender. The attacker's view of the system state will change if it successfully finds a sensitive target based on its target discovery strategy $\sigma$, while the defender's view gets updated as a consequence of discovering an adversarial move (whether on network levels with probability $\gamma$ or on host levels with probability $h$ via a defense agent). 
Algorithm \ref{alg:sim} shows the above description of the game in which the game simulator takes $\mathcal{N,K,H},k$, and an adversarial target discovery strategy $\sigma$, in addition to the sampling rate $\gamma$, and a ``threshold'' $L^*$ as the input and outputs if the attacker wins the game. The termination rule is whether the attacker finds its desired number of sensitive targets $k$ or if the defender reallocates its resources as it discovers enough amount of attack evidence.

%

%% file: sections/conclusion.tex
This presented work is an attempt at constructing a theory of security and developing a general framework for modeling cyber attacks prevalent today, including opportunistic, targeted and multi-stage attacks while taking practical constraints and observations into consideration. 
To this end, we have modeled the interactions of an adversary (and possibly its agents) with a defensive system during the lifecycle of an attack as an incremental online learning game. In comparison with the available works in this area, which are too simple, specific and static to be used in almost any practical situation, our presented framework, to the best of our knowledge, is the most comprehensive and realistic one which can be used to represent the dynamic interplay between the attacker and the defender. Unlike most of the available research in this area which ignores entirely one player's actions and strategies (usually the defender), we have shown how the game evolves by taking both parties set of available actions and strategies into consideration. More specifically, instead of considering a ``dummy defender'' in our modeling, we gave it the opportunity to learn regarding the attack technology incrementally. As time elapses, and the defender captures more attack samples, it can reach better detection rates and accuracy. This learning rate indeed reflects into higher quality attack signatures and detection rates which can be used to bring a next adversarial move to a halt and hence to contain the adversary meaning that the adversary's probability of making a progressive move in the game decreases consequently.


By focusing on the most significant and tangible aspects of sophisticated cyber attacks i.e., (1) the amount of time it takes for the adversary to accomplish its mission and (2) the success probabilities of fulfilling the attack objectives, we were able to study under which circumstances the defense system can provide an effective response that makes the probability of  reaching an attack objective $k$ to be $poly(T_{bud})negl(k)$ in which $T_{bud}$ is the attacker time budget. This led us to the definition of a \emph{security capacity region} as a metric for gauging a defensive system’s efficiency from a security perspective. In particular, we show that a stagnating learning rate allows the attacker to win meaning that being able to reach its attack objective within a limited budget (e.g., time, US dollars), whether it is constructing a botnet of any specific size or locating information on a distributed number of nodes. The defender cannot wait for too long learning about a used attack vector/exploit, and once learning starts it must continue learning with an associated detection probability converging fast enough to $1$. Our security analysis gives precise recommendations for the defender, i.e.,  $1-f(l)\leq O(l^{-a})$ for some $a \geq 2$ with a proof for $a=2$ in our framework. The attacker needs to find just one attack vector/exploit for which the defender is too slow to react or too slow in learning.


An essential venue of future work is to estimate the learning rate $f(l)$ based on the number of observed malicious acts given a ``worst-case adversary''. Our framework lays the foundation for such work and allows to give a worst-case probabilistic bound on the maximal reached attack objective based on the estimated $f(l)$. This, in turn, will give guidance to the defender in how to allocate its resources. In addition, we notice that the learning rate might not always be positive, for instance, when dealing with a \emph{delusive} adversary who maliciously engineers the training data to prevent a learner from generating an accurate classifier, even if the training data is correctly labeled \cite{newsome2006paragraph}. Therefore, in case of noise injection attacks, such as deliberately crafted attack samples to mislead the defender's learning engine, and in general a delusive adversary, the learning engine's false positive rates should be taken into consideration, and we leave this problem for future studies.

%% file: sections/appendix.tex

\subsection{Security Analysis of Case study \ref{CS1}}
 As in section \ref{sec:secAna.sim}, we investigated the role of the convergence rate of $f(.)$ on the attacker's chances of winning the game, we know that for a stagnating learning rate, i.e.,  $f(.) = 1-\tau$, for some $\tau>0$, the adversary will always win. Therefore, $f$ must not stagnate unless the attacker decides to drop out of the game.
To show how effective a learning mechanism could be with respect to containing the adversary's progress in the game, we consider the infamous \emph{CodeRed1v2} worm's actual attack settings and parameters \cite{rohloff2005stochastic} as an example. 
In the codeRed1v2 epidemic, we consider the address space to be $N = 2^{32}$ (the entire IPv4 address space), the approximate number of nodes susceptible to the malware as $K=350,000$ (i.e., $p=8.15 \times 10^{-5}$), and the number of scans performed by an infected machine to be $10,188$ scans per hour, and an initially one infected node at time zero. 

Fig. \ref{fig.Sim_ODE_Learning} compares the total number of infections in the first few hours of the epidemic (1) using the well-accepted simple deterministic epidemic model \cite{zou2002code} and (2) by simulating the game using Algorithm \ref{alg:sim} with no learning. Notice that this plot shows simulations and in this sense, it depicts the average case; no information about the probability of a worst-case can be extracted. The plot indicates that giving the attacker the advantage of $p= \sup_i p_i$ in our analysis is a reasonable assumption as it does not boost its progress in the game significantly.  

Based on the simple deterministic epidemic model (see Fig. \ref{fig.comparison} (a)), we know that the total number of infected machines reaches its maximum (i.e., 350,000) in less than 30 hours, while with a learning function of the form $f(l) = 1 - 1/(l/10,000+1)$, the attacker has control of less than $350$ nodes (on average) and its progress is almost contained due to high value detection rate of the defender (see Fig. \ref{fig.comparison} (b)). 
 

\begin{figure}
	\begin{center}
		\scalebox{0.45}
		{\includegraphics{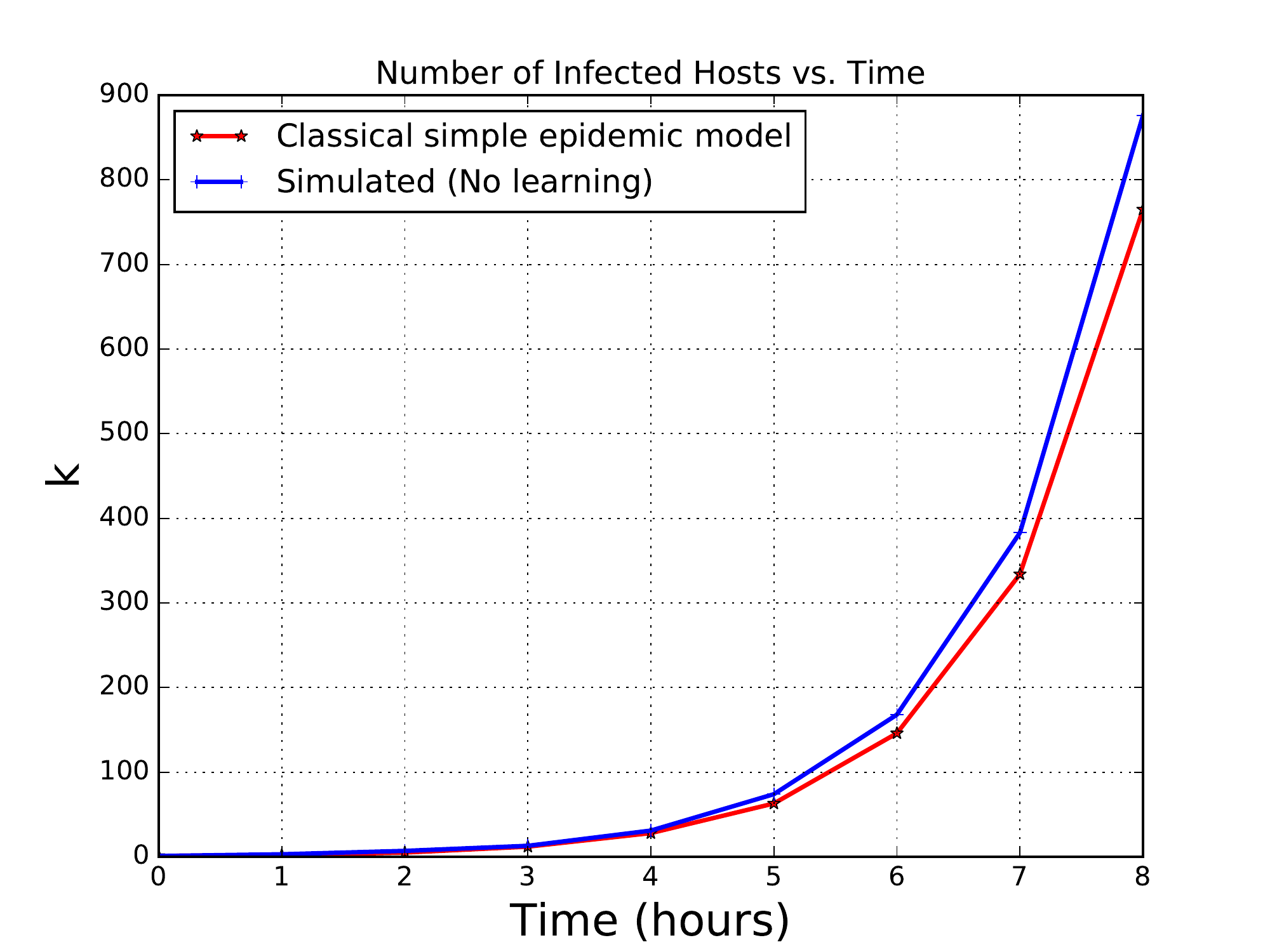}}
		\caption{Studying the role of $p$ on reaching attack objective based on simple deterministic epidemic model and simulations of Algorithm \ref{alg:sim} with no learning, and $p= \sup_i p_i$}
		\label{fig.Sim_ODE_Learning}
	\end{center}
\end{figure}

\begin{figure}
	\begin{center}
		\scalebox{0.45}
		{\includegraphics{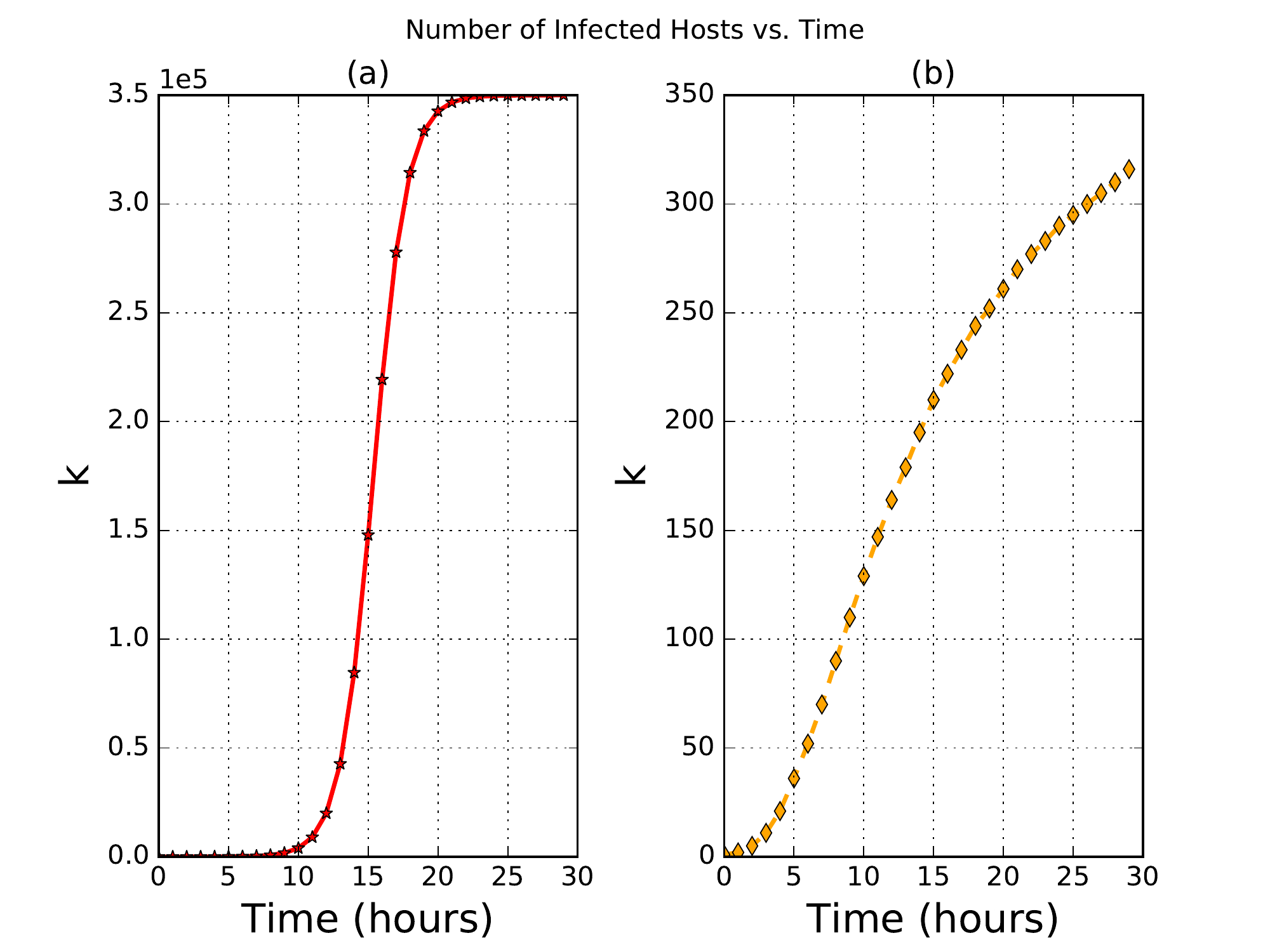}}
		\caption{$(a)$ Simple deterministic epidemic model with the actual codeRed1v2 parameters. $(b)$ Simulation of the codeRed1v2 with the same parameters while the learning function is of the form $f(l)= 1- 1/(l/10,000 +1)$ and $\gamma = 0.01$ as the classifier's accuracy}
		\label{fig.comparison}
	\end{center}
\end{figure}

Let us assume that the defender is interested in figuring out the highest number of infected nodes during a long window of time (e.g., a month) for a designated insignificant  (e.g., $2^{-128} \approx 10^{-38}$) probability of reaching the attack objective by the adversary. 
Fig. \ref{fig.lastOne} depicts $\bar{w}(k,T_{bud})$ (as an upper bound of $w(k,T_{bud})$) when considering a fixed learning function of the form $f(l)= 1- 1/(l/1000 +1)$, while $p=8.15 \times 10^{-5}$, and $\gamma = 0.05$ as the classifier's accuracy. 
Table \ref{table_extra} shows how a simple extrapolation of $k$ can be done for a time window of a month (i.e., $T_{bud} \leq 31 \times 24 \times 10,188$ adversarial moves for the CodeRed example), and setting $\bar{w}$ to $10^{-38}$. 
By defining $k(t)$ as the `fractional' $k$ when $10^{-38}$ is reached  for different $t$ of the form  $T_{bud}=2^t$, it can be seen that the  differences i.e., $k(t+1)-k(t)$ is decreasing for $t > 11$. We can conclude that by taking the last difference, say $\kappa$, and computing the $\log_2$ of the one month time budget $t'$,  the value of $k$ is computed to be at most 130 infected nodes using $k= (t'-14)\cdot \kappa + k(14)$.
This means that in practice, with these attack-defense parameters, {\em the probability that the attacker constructs an army of 130 infected nodes within a one-month time window would be at most $2^{-128}$ } which depicts how vital the learning function is when it comes to containing the adversary. 


\begin{figure}
	\begin{center}
		\scalebox{0.45}
		{\includegraphics{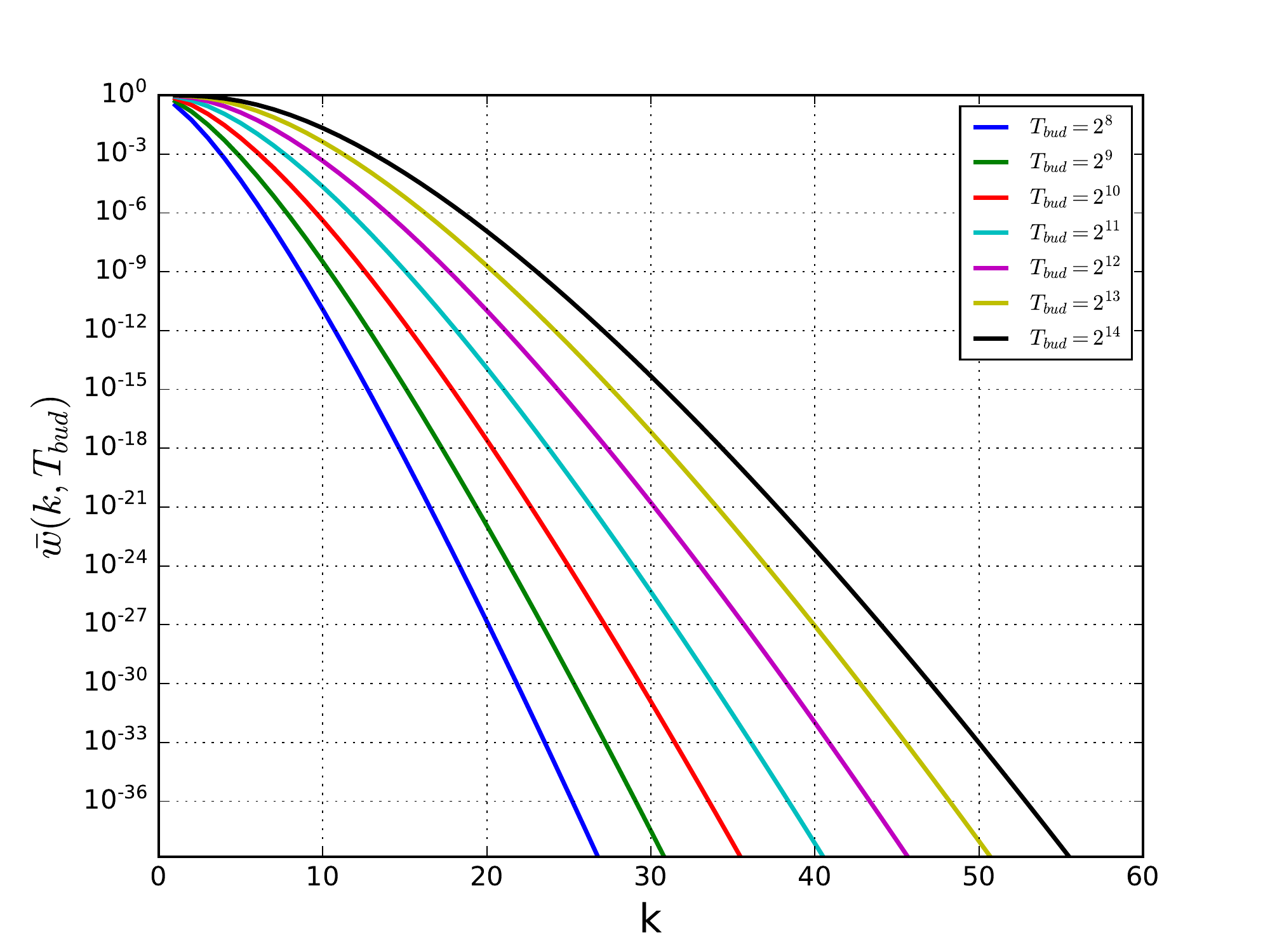}}
		\caption{$\bar{w}(k,T_{bud})$ for different values of $T_{bud} = 2^{t}, 8\leq t \leq 14$ with a fixed learning function of the form $f(l)= 1- 1/(l/1000 +1)$, $p=8.15 \times 10^{-5}$, and $\gamma = 0.05$ as the classifier's accuracy}
		\label{fig.lastOne}
	\end{center}
\end{figure}

\begin{table}[!t]
\renewcommand{\arraystretch}{1.3}
\caption{Extrapolating $k$ for a One-month Time Window and $\bar{w}\approx 10^{-38}$}
\label{table_extra}
\centering
\begin{tabular}{ |p{2cm}|p{2cm}|p{2cm}|  }
	\hline
	 $t~~ (T_{bud}=2^t)$ &$k(t)$&$k(t+1)-k(t)$\\
	\hline
	8    &25.61&   4.01\\
	9  & 29.62   &4.64\\
	10 & 34.26&  5.03\\
	11 & 39.29&  5.13\\
	12  & 44.42&5.02\\
	13  & 49.44   &4.78\\
	14  & 54.22&-\\
	\hline
\end{tabular}
\end{table}

\subsection{Security Analysis of Case study \ref{CS2}}

In this case, the defender is not trying to generate an attack signature and is therefore not able to bring a next adversarial move to a halt. On the other hand, the defender defense strategy is to reallocate/shuffle its resources (or at least a fraction of sensitive ones) when a sufficient number of attack evidence is captured via the defense system (or in other words, the defender's objective is to generate a proper attack detection flag). 
Notice that the presented analysis in section \ref{AC-phase} can immediately be applied to this case study. As $f(l)$ remains $0$ during this game, meaning the defender is not capable of taking actions to bring an adversarial move to a halt, therefore, the delayed learning analysis (see Theorem \ref{Theo.delay}) can be used to study how many targets will be discovered during the game by the adversary, before a threshold $L^*$ is reached by the defender after which the defender shuffles its resources pushing the adversary to the beginning state of the game in the Markov model. The probability that $> k$ targets with unlimited time budget are discovered is equal to $1-u(k,L^*)$ which is exponentially small in $k$. (A more exact upper bound for limited time budgets can be found by applying Theorem \ref{Theo.wbar}.)

%% file: sections/appendix_2.tex

\subsection{Stagnating Learning Rate -- Proof of Theorem \ref{Theo.stag}}

In order to prove a lower bound on $w(k,T_{bud})$ we give a benefit to the defender and assume the learning rate $f(l)$ is as large as possible given  $1-f(l)\geq\tau$, i.e., $1-f(l)=\tau$ with equality.

If $1-f(l) = \tau$, then $m_1,m_2,m_3$ and $m_4$ are independent of $l$ and the Markov model reduces to Fig. \ref{fig.reducedMarkovModel} where each  state $i$ represents the {\em collection} of states $(i,l)$ for $l\geq 0$ in the original Markov model.

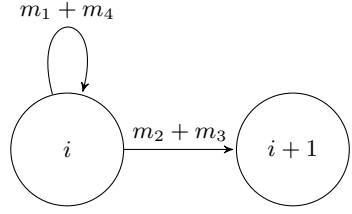
\begin{figure}
	\centering
	\begin{tikzpicture}
	\tikzset{
		>=stealth',
		node distance=3cm and 3cm ,on grid,
		every text node part/.style={align=center},
		state/.style={minimum width=1.5cm,
			draw,
			circle},
	}
	\node[state]             (c) {$i$};
	\node[state, right=of c] (r) {$i+1$};
	\draw[every loop]
	(c) edge[auto=left]  node {$m_{2}+m_3$} (r)
	(c) edge[loop above]             node {$m_1+m_4$} (c);
	\end{tikzpicture}
	\caption{Reduced Markov model of the game for a constant $1-f(l)=\tau>0$}
	\label{fig.reducedMarkovModel}
\end{figure}

Within time budget $T_{bud}$ we reach $i=k$ with probability 
%
	$$ w(k,T_{bud})=  \sum_{T=0}^{T_{bud}-k} {\sum_{\substack{f_0,\dots,f_{k-1} s.t.\\ \sum f_j=T}}
		\prod_{j} (m_1+m_4)^{f_j}(m_2+m_3)^{k}
	}, $$
where the time budget is distributed over $k$ steps from each state to the next until state $k$ is reached and  $T\leq T_{bud}-k$ transitions consisting of $f_i$ self-loops in state $i$ for $0\leq i\leq k-1$.
This probability is equal to
%
%
%
	\begin{eqnarray*}
		&& \sum_{\substack{f_0,\dots,f_{k-1} s.t.\\ \sum f_j\leq T_{bud}-k}}
		\prod_{j} (m_1+m_4)^{f_j}(m_2+m_3)^{k} \\
		&\geq& \sum_{f_0=0}^{(T_{bud}-k)/k}\dots\sum_{f_{k-1}=0}^{(T_{bud}-k)/k}\prod_{j}(m_1+m_4)^{f_j}(m_2+m_3)^{k}\\
		&=& \prod_{j=0}^{k-1}\sum_{f_j=0}^{(T_{bud}-k)/k}(m_1+m_4)^{f_j}(m_2+m_3)^{k}\\
		&=& \prod_{j=0}^{k-1}\frac{1-(m_1+m_4)^{T_{bud}/k}}{1-(m_1+m_4)}(m_2+m_3)^{k}\\
		&=& [1-(m_1+m_4)^{T_{bud}/k}]^{k} \geq 1-k(m_1+m_4)^{T_{bud}/k},
	\end{eqnarray*}
where the last equality follows from $1-(m_1+m_4)=m_2+m_3$. Now we substitute $m_1+m_4=1-\tau p$ which yields Theorem \ref{Theo.stag}. 


\subsection{Closed Form Winning Probability -- Proof of Theorem \ref{closedw}}

Before proving any bounds on $w(k,T_{bud})$ we first provide a closed form for $w(k,T_{bud})$ itself. To this purpose we introduce $p(k-1,L,T_{bud})$ defined as {\em the probability that the adversary reaches state $(k-1,L)$ and is ready to leave state $(k-1,L)$ (after zero or more self-loops at state $(k-1,L)$) within time budget $T_{bud}$}. Any possible path from state $(0,0)$ to state $(k-1,L)$ in the Markov model with self-loops along the way in each of the visited states (including $(k-1,L)$) contributes to this probability. In our notation we use $k-1$ rather than $k$ because after a path to $(k-1,L)$ for some $L\geq 0$ only a single horizontal or diagonal move is needed to reach $(k,L)$ for the first time, and this probability is what we will need for our characterization of $w(k,T_{bud})$:

The formula for $w(k,T_{bud})$ considers all the possible paths towards a state $(k-1,L)$ within $T_{bud} -1$ transitions and a single final horizontal or diagonal transition (via $m_2(L)$ or $m_3(L)$) towards the winning state $(k,L)$. We are interested in the probability of entering state $(k,L)$ for some $L\geq 0$ and the time it takes to reach there \emph{for the first time}. Appropriately summing over probabilities $p(k-1,L,T_{bud}-1)$ gives

\begin{eqnarray}
w(k, T_{bud})
 &=&	\sum_{L\geq 0} p(k-1,L,T_{bud}-1) \cdot	 \bigg( \frac{m_2(L)}{1-m_1(L)}+\frac{m_3(L)}{1-m_1(L)}\bigg), \label{firstw}
\end{eqnarray}
where each path represented by $p(k-1,L,T_{bud})$ is ready to leave $(k-1,L)$ implying that the horizontal and diagonal transition probabilities are conditioned on ``not having a self-loop'' and this explain the division by $1-m_1(L)$.

A path to $(k-1,L)$ till the moment it is ready to leave $(k-1,L)$ is uniquely represented by 
\begin{itemize}
\item $g_l$; the number of states on the path that have the same level $l$,
\item $b_l$; if level $l$ is entered via a vertical transition, then $b_l=0$; if level $l$ is entered via a diagonal transition, then $b_l=1$,
\item $t(i,l)$ counts the number of self-loops on the path in state $(i,l)$.
\end{itemize}
Notice that $b_0$ is undefined and $g_l\geq 1$ for all $0\leq l \leq L$. Fig. \ref{fig.samplePath} depicts a typical path from $(0,0)$ to $(k-1,L)$ without showing any self-loops, where 
\begin{itemize}
\item $i_{l+1}=i_l+b_{l+1}+g_{l+1}$ and $i_0=0$; when the path enters level $l$ for the first time, $k=i_l$.
\end{itemize}

	\begin{figure}
		\begin{center}
			\scalebox{0.50}
			{\includegraphics{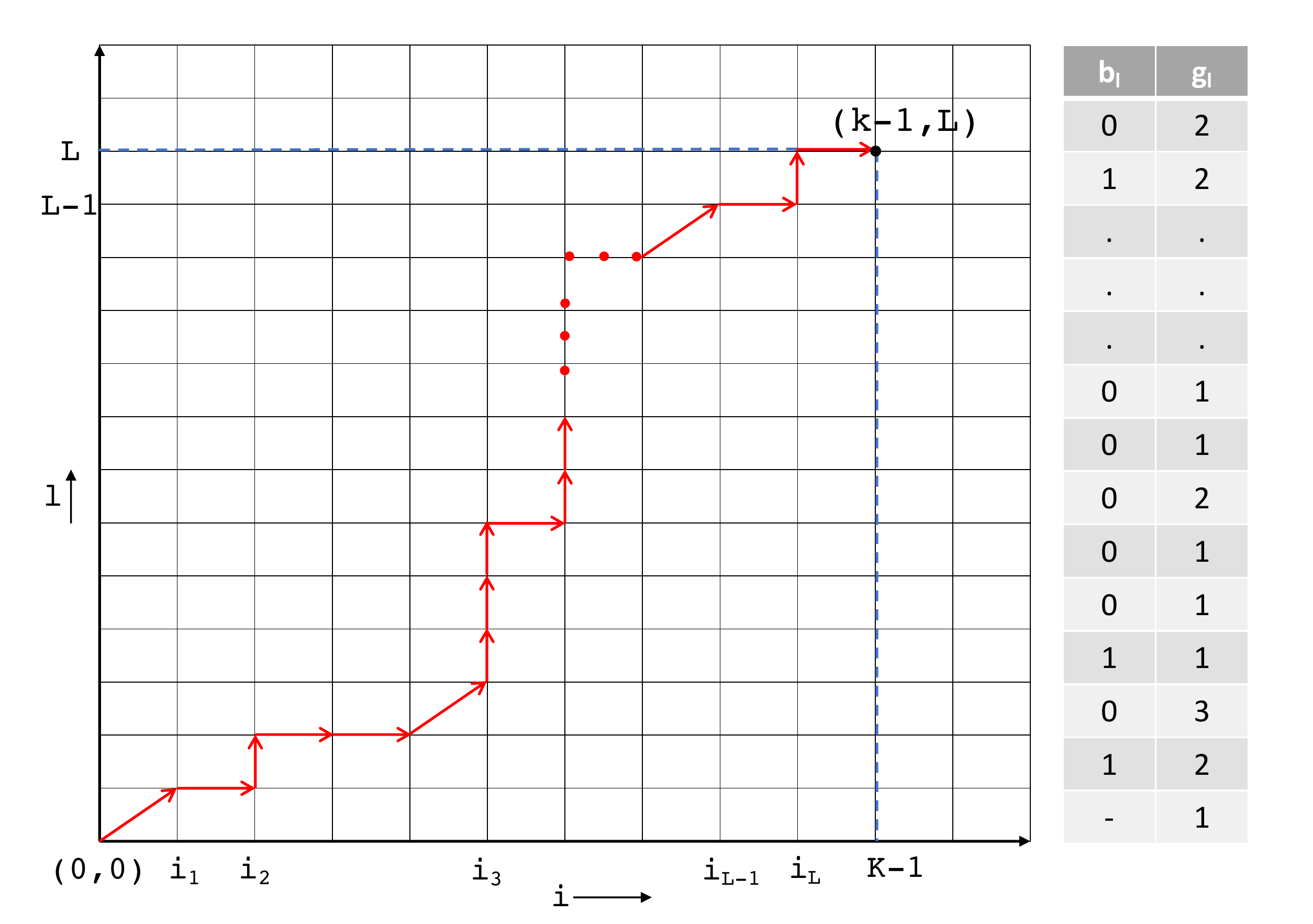}}
			\caption{A sample path toward the winning state without depicting the self-loops at each state $(i,l)$ in the path}
			\label{fig.samplePath}
		\end{center}
	\end{figure}
	
	The probability of having exactly $t(i,l)$ self-loops in state $(i,l)$ after which the path exits $(i,l)$ is equal to 
	$$Prob(t(i,l)=u) = m_1(l)^{u}(1-m_1(l))$$
	and describes a Poisson process.
	The probability  that the path exits $(i,l)$ along a horizontal, diagonal or vertical transition given  no more self-loops in $(i,l)$ is equal to
	$$
	\begin{aligned}
	Prob(\mbox{horizontal~transition}) &= Prob(\mytiny{(}i,l\mytiny{)}\rightarrow\mytiny{(}i+1,l\mytiny{)}|\mbox{no~self-loop})\\
	&= \frac{m_2(l)}{1-m_1(l)},\\
	Prob(\mbox{diagonal~transition}) &= Prob(\mytiny{(}i,l\mytiny{)}\rightarrow\mytiny{(}i+1,l+1\mytiny{)}|\mbox{no~self-loop})\\
	&= \frac{m_3(l)}{1-m_1(l)},~\mbox{and}\\
	Prob(\mbox{vertical~transition}) &= Prob(\mytiny{(}i,l\mytiny{)}\rightarrow\mytiny{(}i,l+1\mytiny{)}|\mbox{no~self-loop})\\
	&= \frac{m_4(l)}{1-m_1(l)}.\\
	\end{aligned}
	$$
Combination of the probabilities describing diagonal and vertical transitions shows that
$$ Prob(\mytiny{(}i_{l}+b_{l+1}+g_{l+1},l\mytiny{)}\rightarrow\mytiny{(}i_{l+1},l+1\mytiny{)}|\mbox{no~self-loop})
=  b_{l+1}\frac{m_3(l)}{1-m_1(l)} + (1-b_{l+1})\frac{m_4(l)}{1-m_1(l)}.$$
The above analysis proves
that the probability of having the Markov model transition along a path which is represented by $g_l,b_l,t(i,l)$, and $i_l$ is equal to
		\begin{subequations}
			\begin{eqnarray}
			&& \prod_{l=0}^{L-1} \bigg( b_{l+1}\frac{m_3(l)}{1-m_1(l)} + (1-b_{l+1})\frac{m_4(l)}{1-m_1(l)} \bigg) \label{path-a}\\
			&&  \cdot \prod_{l=0}^{L} \bigg( \frac{m_2(l)}{1-m_1(l)}\bigg)^{g_l-1} \label{path-b}\\
			&&  \cdot \prod_{l=0}^{L}\prod_{i=0}^{g_l-1}m_1(l)^{t(i_l+i,l)}(1-m_1(l)), \label{path-c}
			\end{eqnarray}
	\end{subequations}
where (\ref{path-a}) corresponds to diagonal and vertical transitions, (\ref{path-b}) corresponds to horizontal transitions, and (\ref{path-c}) corresponds to self-loops.

The length of the path, i.e., total number of transitions without self-loops,  is equal to 
$$ 
(k-1) + L -B, \mbox{ where } B=\sum_{l=1}^L b_l
$$
is the total number of diagonal transitions.  The total number of vertical transitions is equal to $L-B$ and the number of horizontal transitions is equal to $\sum_{l=0}^L (g_l-1) = (k-1)-B$.

Summing the above probability over all possible combinations of $g_l$, $b_l$, and $t(.,.)$ for paths that reach state $(k-1,L)$  and are ready to leave $(k-1,L)$ such that the number of self-loops is equal to $T$ and the total number of transitions including self-loops is equal to $(k-1)+L-B + T\leq T_{bud}$ leads to an expression for $p(k-1,L,T_{bud})$:

	\begin{subequations}
		\begin{eqnarray}
	&&	p(k-1,L,T_{bud}) = 
		\sum_{\substack{B,T s.t.\\
				(k-1)+L-B+T\leq T_{bud}}}\nonumber\\& &
				\hspace{1cm}
		{ 
			\sum_{\substack{b_1,\dots,b_L\in\{0,1\} s.t.\\ 
					\sum_{l=1}^{L}b_l=B}}
(\ref{path-a})
		} \label{p-a} \\& &
		\hspace{1cm}
		{	\sum_{\substack{g_0\geq 1,\dots,g_L\geq 1 s.t.\\ \label{p-b}
					\sum_{l=0}^{L}(g_l-1)=(k-1)-B}}
(\ref{path-b})
		}   \label{p-b} \\& &
		\hspace{1cm}
		{
			\sum_{\substack{t_0,\dots,t_L s.t.\\ 
					\sum_{l=0}^{L}t_l=T}}
			\sum_{\substack{t(i_0,0),\dots,t(i_0+g_0-1,0) \\
						s.t.\sum_{i=0}^{g_0-1}t(i_0+i,0)=t_0}}
						\ldots
	\sum_{\substack{t(i_L,L),\dots,t(i_L+g_L-1,L) \\
						s.t.\sum_{i=0}^{g_L-1}t(i_L+i,L)=t_L}}					
(\ref{path-c})} \label{p-c}
		\end{eqnarray}
	\end{subequations}
	
	Notice that $(\ref{path-a})=0$ if $B>L$ and $(\ref{path-b})=0$ if $B>k-1$. So, the main sum only needs to consider $B\leq \min\{L,(k-1)\}$.
	Plugging the above formula into (\ref{firstw}) gives

	\begin{subequations}
\begin{eqnarray}
w(k, T_{bud})
&=&  \sum_{L=0}^{T_{bud}-1}\sum_{B=0}^{\min\{L,(k-1)\}}
  \frac{m_2(L)+m_3(L)}{1-m_1(L)} \cdot (\ref{p-a})\cdot (\ref{p-b}) \label{w-a} \\
&&
 \cdot \sum_{T=0}^{(T_{bud}-1)-[(k-1)+L-B]} (\ref{p-c}).  \label{w-b}
\end{eqnarray}
\end{subequations}

This  closed form together with
{\footnotesize
\begin{eqnarray*}
 (\ref{p-c}) &=& 
	\sum_{\substack{t_0,\dots,t_L s.t.\\ 
					\sum_{l=0}^{L}t_l=T}}
			\sum_{\substack{t(i_0,0),\dots,t(i_0+g_0-1,0) \\
						s.t.\sum_{i=0}^{g_0-1}t(i_0+i,0)=t_0}}
						\ldots
	\sum_{\substack{t(i_L,L),\dots,t(i_L+g_L-1,L) \\
						s.t.\sum_{i=0}^{g_L-1}t(i_L+i,L)=t_L}}					
\prod_{l=0}^{L}			
\prod_{i=0}^{g_l-1}m_1(l)^{t(i_l+i,l)}(1-m_1(l)) 
\end{eqnarray*}
\begin{eqnarray*}
&=&
\sum_{\substack{t_0,\dots,t_L s.t.\\ 
					\sum_{l=0}^{L}t_l=T}}
			\sum_{\substack{t(i_0,0),\dots,t(i_0+g_0-1,0) \\
						s.t.\sum_{i=0}^{g_0-1}t(i_0+i,0)=t_0}}
						\ldots
	\sum_{\substack{t(i_L,L),\dots,t(i_L+g_L-1,L) \\
						s.t.\sum_{i=0}^{g_L-1}t(i_L+i,L)=t_L}}	
\prod_{l=0}^{L}		
 m_1(l)^{t_l}(1-m_1(l))^{g_l}	\\
 &=&
\sum_{\substack{t_0,\dots,t_L s.t.\\ 
					\sum_{l=0}^{L}t_l=T}}		
\prod_{l=0}^{L}	 \sum_{\substack{t(i_l,l),\dots,t(i_l+g_l-1,l) \\
						s.t.\sum_{i=0}^{g_l-1}t(i_l+i,l)=t_l}}
 m_1(l)^{t_l}(1-m_1(l))^{g_l}	\\
&=&	
	\sum_{\substack{t_0,\dots,t_L s.t.\\ 
					\sum_{l=0}^{L}t_l=T}}
	\prod_{l=0}^{L}			{t_l+g_l-1 \choose t_l}
 m_1(l)^{t_l}(1-m_1(l))^{g_l}	
\end{eqnarray*}						
}
proves Theorem \ref{closedw}.



\subsection{Upper and Lower Bounds -- Proof of Theorem \ref{Theo.wbar}}

By noticing that (\ref{w-b}) for $T_{bud}=\infty$ is equal to 1, a straightforward upper bound on $w(k,T_{bud})$ is given by just formula (\ref{w-a}): We have $w(k,T_{bud})$ is at most equal to $\bar{w}(k,T_{bud})$.


In order to prove a lower bound we first analyze
\begin{eqnarray*}
(\ref{w-b}) &=& \sum_{T=0}^{(T_{bud}-1)-[(k-1)+L-B]} (\ref{p-c})\\
&=& \sum_{T=0}^{(T_{bud}-k-L+B)}\sum_{\substack{t_0,\dots,t_L s.t.\\
		\sum_{l=0}^{L}t_l=T}} 
\sum_{\substack{t(i_0,0),\dots,t(i_0+g_0-1,0) \\
						s.t.\sum_{i=0}^{g_0-1}t(i_0+i,0)=t_0}}
						\ldots
	\sum_{\substack{t(i_L,L),\dots,t(i_L+g_L-1,L) \\
						s.t.\sum_{i=0}^{g_L-1}t(i_L+i,L)=t_L}} \\
						&& \hspace{2cm}
 \prod_{l=0}^{L}	\prod_{i=0}^{g_l-1}m_1(l)^{t(i_l+i,l)}(1-m_1(l))\\
&=& \sum_{\substack{t_0,\dots,t_L s.t.\\
		\sum_{l=0}^{L}t_l \leq (T_{bud}-k-L+B)}} 
\sum_{\substack{t(i_0,0),\dots,t(i_0+g_0-1,0) \\
						s.t.\sum_{i=0}^{g_0-1}t(i_0+i,0)=t_0}}
						\ldots
	\sum_{\substack{t(i_L,L),\dots,t(i_L+g_L-1,L) \\
						s.t.\sum_{i=0}^{g_L-1}t(i_L+i,L)=t_L}} \\
						&& \hspace{2cm}
 \prod_{l=0}^{L} \prod_{i=0}^{g_l-1}m_1(l)^{t(i_l+i,l)}(1-m_1(l)).
 \end{eqnarray*}

 Notice that (\ref{w-b}) is used within another sum with $B\leq \min\{L,(k-1)\}$. This implies that $B\leq (k-1)L$ or equivalently $-(L+1)k\leq -k-L+B$.
 Let $v\geq 1$. For now we only consider $L\leq (T_{bud}/v)-1$, or equivalently $v\leq T_{bud}/(L+1)$. Then 
 $(L+1)(v-k)\leq T_{bud}-(L+1)k\leq T_{bud}-k-L+B$, hence,
 $$v-k\leq \frac{T_{bud}-k-L+B}{L+1}.$$
 This allows us to lower bound the above sums and obtain
 \begin{eqnarray*}
(\ref{w-b})  &\geq& \sum_{t_0=0}^{v-k}\dots\sum_{t_L=0}^{v-k}  
\sum_{\substack{t(i_0,0),\dots,t(i_0+g_0-1,0) \\
						s.t.\sum_{i=0}^{g_0-1}t(i_0+i,0)=t_0}}
						\ldots
	\sum_{\substack{t(i_L,L),\dots,t(i_L+g_L-1,L) \\
						s.t.\sum_{i=0}^{g_L-1}t(i_L+i,L)=t_L}}\\
						&& \hspace{2cm}
\prod_{l=0}^{L}\prod_{i=0}^{g_l-1}m_1(l)^{t(i_l+i,l)}(1-m_1(l))\\
&\geq&
\sum_{\substack{t(i_0,0),\dots,t(i_0+g_0-1,0) \\
						s.t.\sum_{i=0}^{g_0-1}t(i_0+i,0)\leq v-k}}
						\ldots
	\sum_{\substack{t(i_L,L),\dots,t(i_L+g_L-1,L) \\
						s.t.\sum_{i=0}^{g_L-1}t(i_L+i,L)\leq v-k}}
\prod_{l=0}^{L}\prod_{i=0}^{g_l-1}m_1(l)^{t(i_l+i,l)}(1-m_1(l))\\
&\geq& 
\sum_{t(i_0,0)\leq t_0/g_0}\dots \sum_{t(i_0+g_0-1,0) \leq (v-k)/g_0}
\ldots
\sum_{t(i_L,L)\leq t_L/g_L}\dots \sum_{t(i_L+g_L-1,L) \leq (v-k)/g_L} \\
&& \hspace{2cm}
 \prod_{l=0}^{L} \prod_{i=0}^{g_l-1}m_1(l)^{t(i_l+i,l)}(1-m_1(l))\\
&=& \prod_{l=0}^{L} \prod_{i=0}^{g_l-1} \sum_{t(i_l+i,l)\leq(v-k)/g_l}m_1(l)^{t(i_l+i,l)}(1-m_1(l))\\
&=& \prod_{l=0}^{L} \prod_{i=0}^{g_l-1}\frac{1-m_1(l)^{\frac{v-k}{g_l} +1}}{1-m_1(l)}(1-m_1(l))\\
&=&\prod_{l=0}^{L}(1-m_1(l)^{\frac{v-k}{g_l} +1})^{g_l} \geq  \prod_{l=0}^{L}(1-g_l m_1(l)^{\frac{v-k}{g_l} +1}).
\end{eqnarray*}

We observe that $g_l\leq k$ within the larger sum that contains  (\ref{w-b}). If $v\geq k$, then this implies $v/k \leq 1+(v-k)/g_l$ proving
 \begin{eqnarray*}
(\ref{w-b})
&\geq& \prod_{l=0}^{L}(1-g_l m_1(l)^{v/k})
\end{eqnarray*}

%
Since $m_1(l)=[(1-\gamma)(1-(p+h)(1-f(l)))]\leq 1-\gamma$,  we can further lower bound this to
\begin{eqnarray*}
(\ref{w-b})	&\geq& \prod_{l=0}^{L}(1-g_l (1-\gamma)^{v/k})\\
	&\geq& 1- (\sum_{l=0}^{L}g_l)(1-\gamma)^{v/k}. 
\end{eqnarray*}
Within the larger sum that contains (\ref{w-b}), we have $\sum_{l=0}^{L} g_l = (L+1) + \sum_{l=0}^L (g_l-1)= (L+1)+(k-1-B)=k + L-B\leq k+T_{bud}/v$. This yields 
$$(\ref{w-b})\geq 1-(k+\frac{T_{bud}}{v})(1-\gamma)^{v/k}.$$

The obtained lower bound on $(\ref{w-b})=\sum_{T=0}^{(T_{bud}-1)-[(k-1)+L-B]} (\ref{p-c})$  is independent of any of the other summing variables used in (\ref{w-a}) but requires $L+1\leq T_{bud}/v$. By restricting $L$ to be $\leq (T_{bud}/v)-1$ in (\ref{w-a}), i.e., we substitute $T_{bud}$ by $T_{bud}/v$, we obtain a lower bound on $w(k,T_{bud})$:
\begin{eqnarray*}
w(k, T_{bud})
&=&  \sum_{L=0}^{T_{bud}-1}\sum_{B=0}^{\min\{L,(k-1)\}} 
  \frac{m_2(L)+m_3(L)}{1-m_1(L)} \cdot (\ref{p-a})\cdot (\ref{p-b}) 
 \cdot \sum_{T=0}^{(T_{bud}-1)-[(k-1)+L-B]} (\ref{p-c}) \\
 &\geq &  \sum_{L=0}^{(T_{bud}/v)-1}\sum_{B=0}^{\min\{L,(k-1)\}} 
  \frac{m_2(L)+m_3(L)}{1-m_1(L)} \cdot (\ref{p-a})\cdot (\ref{p-b}) 
 \cdot \sum_{T=0}^{(T_{bud}-1)-[(k-1)+L-B]} (\ref{p-c}) \\
& \geq & \sum_{L=0}^{T_{bud}/v-1}\sum_{B=0}^{\min\{L,(k-1)\}} 
  \frac{m_2(L)+m_3(L)}{1-m_1(L)} \cdot (\ref{p-a})\cdot (\ref{p-b}) (1-(k+\frac{T_{bud}}{v})(1-\gamma)^{v/k}) \\
  &=&
 \bar{w}(k,T_{bud}/v)\cdot (1-(k+\frac{T_{bud}}{v})(1-\gamma)^{v/k}),
\end{eqnarray*}
where the last equality follows from the fact that (\ref{p-a}) and (\ref{p-b}) do not depend on $T_{bud}$ (they depend on $L$). This completes the proof.



\subsection{Analyzing the Learning Rate -- Proof of Theorem \ref{Theo.wbarUB}}

Theorem \ref{Theo.wbarUB} has a couple of statements and we start by proving the first most general claim that upper bounds $\bar{w}(k,T_{bud})$ in terms of the general parameters of the Markov model. 

We define and assume
$$ \beta(l)=\sqrt{\frac{1-f(l)}{d}} \leq 1.$$
This allows us to bound
\begin{eqnarray*}
\alpha(l) &=& \frac{p(1-f(l))}{\gamma + (1-\gamma)(p+h)(1-f(l))} \leq \frac{p(1-f(l))}{\gamma}=\frac{pd\beta(l)^2}{\gamma}, \\
\alpha(l) &=& \frac{p(1-f(l))}{\gamma + (1-\gamma)(p+h)(1-f(l))} \geq p(1-f(l)) = pd\beta(l)^2.
\end{eqnarray*}
Since $\beta(l)\leq 1$,  $\beta(l)^2\leq \beta(l)$.
Let $\theta$ be such that $(1-\gamma)/\gamma\leq \theta$. Then the above inequalities allow us to bound 
	\begin{eqnarray*}
		\frac{m_2(l)}{1-m_1(l)} &=& (1-\gamma)\alpha(l) \leq \frac{(1-\gamma)pd}{\gamma} \beta(l)^2\leq \theta pd \beta(l)^2 \\
		&\leq& \theta pd  \beta(l),\\
		\frac{m_3(l)}{1-m_1(l)} &=& \gamma \alpha(l) \leq  pd \beta(l)^2, \mbox{ and}\\
		\frac{m_4(l)}{1-m_1(l)} &=& 1-\alpha(l) \leq 1- pd \beta(l)^2.
	\end{eqnarray*}

Plugging the above bounds in expression (\ref{UB}) for $\bar{w}(k,T_{bud})$ with $\frac{m_2(L)}{1-m_1(L)}\leq \theta pd \beta(L)^2$ and  $\frac{m_2(l)}{1-m_1(l)}\leq \theta pd \beta(l)$ yields
	\begin{subequations}
		\begin{eqnarray}
		&& \bar{w}(k,T_{bud})
		\leq
		\sum_{L=0}^{T_{bud}-1} (1+ \theta)pd \beta(L)^2 \nonumber \\ 
		& &	
		\cdot \sum_{B=0}^{L} { \sum_{\substack{b_1,\dots,b_L\in\{0,1\} \\ 
					s.t. \sum_{l=1}^{L}b_l=B}}
			\prod_{l=0}^{L-1} \begin{array}{l} b_{l+1}pd \beta(l)^2  \\ + (1-b_{l+1})(1-pd \beta(l)^2) \end{array} 
		} \label{ub-a} \\
		& &
		{	\cdot \sum_{\substack{g_0,\dots,g_L\geq 1 s.t.\\ 
					\sum_{l=0}^{L} (g_l-1)=(k-1)-B}}
			\prod_{l=0}^{L} \bigg(\theta pd \beta(l)\bigg)^{g_l-1} 
		} \label{ub-b}
		\end{eqnarray}
	\end{subequations}

Let $b_{L+1}=0$ and define
	$$ Z = \prod_{l=0}^{L} \bigg(\theta pd \beta(l)\bigg)^{b_{l+1}}.$$

We will multiply (\ref{ub-b}) with $Z$ and show an upper bound of the product which is independent of $b_l$ and $B$. We will multiply (\ref{ub-a}) with $Z^{-1}$ and show that it behaves like a product of $1+\frac{\beta(l)}{\theta}$, which (as we will see) remains `small enough': 

We choose $d$ such that $\theta\leq (pd)^{-1}$ (if $d$ does not satisfy this inequality, then the to be derived upper bound will be larger than 1 and will therefore trivially hold for $\bar{w}(k,T_{bud})$). Then, since $B=\sum_{l=0}^{L}b_{l+1}$,
	\begin{eqnarray}
	&&  (\ref{ub-b})\cdot Z = \sum_{\substack{g_0,\dots,g_L\geq 1 s.t.\\ 
			\sum_{l=0}^{L} (g_l-1)+b_{l+1}=(k-1)}}
	\prod_{l=0}^{L} \bigg(\theta pd \beta(l)\bigg)^{(g_l-1)+b_{l+1}} \nonumber \\
	&& \leq \sum_{\substack{g'_0,\dots,g'_L\geq 0 s.t.\\ 
			\sum_{l=0}^{L} g'_l=(k-1)}}
	\prod_{l=0}^{L} \bigg(\theta pd \beta(l)\bigg)^{g'_l} \nonumber \\
	&&\leq
	(\theta pd)^{k-1}  \sum_{\substack{g'_0,\dots,g'_L\geq 0 s.t.\\ 
			\sum_{l=0}^{L} g'_l=(k-1)}}
	\prod_{l=0}^{L} \beta(l)^{g'_l} \nonumber \\
	&&\leq
	(\theta pd)^{k-1}  \sum_{g'_0\geq 0} \ldots \sum_{g'_L\geq 0}
	\prod_{l=0}^{L} \beta(l)^{g'_l} \nonumber \\
	&&=  (\theta pd)^{k-1} \prod_{l=0}^{L} \sum_{g'_l\geq 0} \beta(l)^{g'_l} 
	=(\theta pd)^{k-1}  \prod_{l=0}^{L} \frac{1}{1-\beta(l)}. \nonumber
	\end{eqnarray}

Now notice that each $b_{l+1}$ is either equal to $0$ or $1$ and therefore
$$(b_{l+1}pd\beta(l)^2+(1-\beta_{l+1})(1-pd\beta(l)^2)\cdot (\theta pd \beta(l))^{-b_{l+1}}
= b_{l+1} \frac{\beta(l)}{\theta} +(1-\beta_{l+1})(1-pd\beta(l)^2.$$
We derive
	\begin{eqnarray}
	&&  (\ref{ub-a})\cdot Z^{-1} 
	= 
	\sum_{B=0}^{L} \sum_{\substack{b_1,\dots,b_L\in\{0,1\} \\ 
			s.t. \sum_{l=1}^{L}b_l=B}}
	\prod_{l=0}^{L-1} \begin{array}{l} b_{l+1}\frac{ \beta(l)}{\theta} \\
	+ (1-b_{l+1})(1-pd \beta(l)^2) \end{array} \nonumber \\
	&& =
	\sum_{b_1,\dots,b_L\in\{0,1\} }
	\prod_{l=0}^{L-1}  \begin{array}{l} b_{l+1}\frac{ \beta(l)}{\theta} \\
	+ (1-b_{l+1})(1-pd \beta(l)^2) \end{array} \nonumber \\
	&&	= 
	\prod_{l=0}^{L-1} \sum_{b_{l+1}\in \{0,1\}} 
	\begin{array}{l} b_{l+1}\frac{ \beta(l)}{\theta} \\
	+ (1-b_{l+1})(1-pd \beta(l)^2) \end{array} \nonumber \\
	&& =
	\prod_{l=0}^{L-1} \bigg( \frac{\beta(l)}{\theta} + (1-pd \beta(l)^2) \bigg)
	\leq \prod_{l=0}^{L-1} \bigg( \frac{\beta(l)}{\theta} + 1 \bigg). \nonumber
	\end{eqnarray}
Combination of the above results proves
	\begin{eqnarray*}
		\bar{w}(k,T_{bud}) &\leq &\nonumber
		\sum_{L=0}^{T_{bud}-1} [\theta p d]^{k}(1+\frac{1}{\theta})\frac{\beta(L)^2}{1-\beta(L)}\prod_{l=0}^{L-1} \frac{1+\beta(l)/\theta}{1-\beta(l)}\\
		&\leq&  	\sum_{L=0}^{T_{bud}-1} [\theta p d]^{k}(1+\frac{1}{\theta})\prod_{l=0}^{L} \frac{1+\beta(l)/\theta}{1-\beta(l)},
	\end{eqnarray*}
where the last inequality uses $\beta(L)^2\leq 1\leq 1+\beta(L)/\theta$.
The argument in the resulting product is equal to $(1+\beta(l)/\theta)/(1-\beta(l)) = 1+(1+\theta^{-1})\beta(l)/(1-\beta(l))$. This completes the proof of the first statement. 

In order to prove the second bound in Theorem \ref{Theo.wbarUB} we define 
$$B'(l)=\frac{\beta(l)}{1-\beta(l)} \mbox{ and } y=1+\theta^{-1}$$
and apply the next lemma. Notice that $B'(l)\geq 0$ and, since $f'(l)\geq 0$ (because learning only increases), 
$$B''(l) = \beta'(l)\frac{2-\beta(l)}{(1-\beta(l))^2} = \frac{-f'(l)}{2\beta(l)} \frac{2-\beta(l)}{(1-\beta(l))^2}\leq 0.$$

\begin{lemma}\label{lemma.diri}
	For differentiable continuous functions $B(.)$, if $B'(l)\geq 0$ and $B''(l)\leq 0$ for $l\geq 0$, then, for $y\geq 0$,
	{\footnotesize
		$$\prod_{l=0}^L (1+ yB'(l))\leq e^{y(B'(0)-B(0)+B(L))}.$$
	}
\end{lemma}
\begin{proof}
	We find an upperbound of $\ln \prod_{l=0}^{L}(1+yB'(l)) = \sum_{l=0}^{L}\ln(1+yB'(l))$. Differentiating w.r.t. $y$ gives
	\begin{eqnarray*}
		\sum_{l=0}^{L} \frac{1}{1+yB'(l)}\cdot B'(l) &\leq& \sum_{l=0}^{L}B'(l) = B'(0)+\sum_{l=1}^{L}B'(l)\\
		&\leq& B'(0) + \int_{l=0}^{L} B'(l) dl\\
		&=& B'(0) + B(L) - B(0)
	\end{eqnarray*}
	We also have
	$$\sum_{l=0}^{L}\ln(1+yB'(l))|_{y=0}=0=y \cdot [B'(0)-B(0)+B(L)]|_{y=0}. $$
	We conclude that $\sum_{l=0}^{L}\ln(1+yB'(l)) \leq y \cdot [B'(0)-B(0)+B(L)]$ for $y\geq 0$. \qed
\end{proof}

For our  $B'(l)$ and $y$, application of the lemma to our upper bound yields 
\begin{eqnarray*}
	&& \bar{w}(k,T_{bud}) \leq
	(1+\theta^{-1})\sum_{L=0}^{T_{bud}-1} [\theta p d]^{k} e^{(1+\theta^{-1})(B'(0)-B(0)+B(L))} \\
	&& \leq T_{bud}\cdot (1+\theta^{-1}) [\theta p d]^{k} e^{(1+\theta^{-1})(B'(0)-B(0)+B(T_{bud}-1))} .
\end{eqnarray*}

Minimizing $\theta^k e^{(1+\theta^{-1})c}$ for
$$c=B(T_{bud}-1)-B(0)+B'(0)$$
gives $\theta = c/k$ if $c/k\geq (1-\gamma)/\gamma$ and gives $\theta=(1-\gamma)/\gamma$ if $c/k\leq (1-\gamma)/\gamma$.
Substituting this in our upper bound  ($e$ denotes the natural number) yields
\begin{eqnarray}\label{eq.conditions}
&& \bar{w}(k,T_{bud}) \leq \nonumber \\
&& \left\{ \begin{array}{l}
T_{bud} (1+\frac{k}{c}) [\frac{c}{k} pd]^k e^{k+c} = T_{bud}(1+\frac{k}{c}) [epd\frac{c}{k}e^{c/k}]^k,
\hspace{0.5cm} \text{if}\ c/k \geq (1-\gamma)/\gamma \\
T_{bud} (1-\gamma)^{-1}  [\frac{1-\gamma}{\gamma}pd]^{k} e^{(1-\gamma)^{-1} c}, 
\hspace{1.5cm} \text{ if}\ c/k \leq (1-\gamma)/\gamma.
\end{array} \right. \label{GenUB}
\end{eqnarray}
The first case of the upper bound is less interesting as $\frac{c}{k}e^{c/k}$ may be large yielding a bad upper bound.
The second case of the upper bound gives the most insight and  proves the second statement of  Theorem \ref{Theo.wbarUB}.

As an example, suppose that
$$1-f(l) \leq \frac{d}{(l+2)^2}.$$
To get an upper bound, we give the adversary the advantage of having the defender play with the smallest possible learning rates. i.e., $1-f(l)=\frac{d}{(l+2)^2}$. This gives $B'(l)=1/(l+1)$, hence, $B(T_{bud}-1)=\ln(T_{bud})$, $B(0)=0$, and $B'(0)=1$. This implies $c=1+\ln(T_{bud})$.
The second case of the upper bound translates into: If
\begin{equation}T_{bud} \leq e^{-1+k (1-\gamma)/\gamma }, \label{cond} \end{equation}
then the probability of winning for the adversary is at most
$$w(k,T_{bud}) \leq \bar{w}(k,T_{bud}) \leq 
\frac{e^{1/(1-\gamma)}}{1-\gamma} \cdot T_{bud}^{1+1/(1-\gamma)} \cdot  [\frac{1-\gamma}{\gamma}pd]^{k} .
$$

This proves the last statement of the theorem.


\subsection{Delayed Learning -- Proof of Theorem  \ref{Theo.delay} }

In practice, we may not immediately start learning at a rate $1-f(l) \approx d/(l+2)^2$. This will only happen after reaching for example $L^*$ samples. During such a first phase the defender is not yet able to increase $f(l)$ and $f(l)$ remains $0$, i.e., $1-f(l) = 1$. The adversary tries to compromise as many, say $k^*$, nodes as possible before reaching $L^*$ levels.

For $1-f(l)= 1$, we have:
\begin{equation*}
\begin{aligned}
m_1(l) &= (1-\gamma)(1-(p+h)),\\ 
m_2(l) &= (1-\gamma)p,\\
m_3(l) &= \gamma p, \mbox{ and}\\
m_4(l) &= [h-\gamma(h+p)]+\gamma=(1-\gamma)h+\gamma(1-p),\\
\end{aligned}
\end{equation*}
which are all independent of $l$. For this reason we use $m_1$, $m_2$, $m_3$, and $m_4$ where suited in our derivations below.

We are interested in $k^*$ as a function of $\epsilon$ for which 
\begin{eqnarray}\label{eq.k}
u(k^*, L^*)&=& \sum_{k=0}^{k^{*}-1} p(k,L^{*}-1,T_{bud}-1)(\frac{m_4(L^*-1)}{1-m_1(L^*-1)}+\frac{m_3(L^*-1)}{1-m_1(L^*-1)}) + \\ && \hspace{1cm} p(k^*,L^{*}-1,T_{bud}-1)\frac{m_4(L^*-1)}{1-m_1(L^*-1)} \geq 1- \epsilon. \nonumber
\end{eqnarray}
Here $u(k^*,L^*)$ is equal to the probability that $\leq k^*$ nodes will be compromised before level $l=L^*$ is reached for the first time: The sum in (\ref{eq.k}) considers all paths reaching a state $(k,L^*-1)$ for some $k<k^*$ after which a single vertical or diagonal transition reaches level $L^*$ for the first time. The additional term considers paths that reach a state $(k^*,L^*-1)$ after which only a vertical transition reaches level $L^*$ without increasing the number of compromised nodes beyond $k^*$.
If (\ref{eq.k}) holds, then the probability of winning is 
\begin{equation}\label{eq.deta}
\leq \epsilon + (1-\epsilon)w(k-k^{*}(\epsilon),T_{bud}),
\end{equation}
where $w(k-k^{*}(\epsilon),T_{bud})$ is defined for learning rate $f(l+L^*)$ as a function of $l$. In other words, during the first phase a node with $i<k^*(\epsilon)$ is reached with probability $\geq 1-\epsilon$ after which the learning rate increases according to $f(l+L^*)$ leading to a winning state if the Markov model increments $i$ at least another $k-k^*(\epsilon)$ steps.


In order to find a lower bound on $u(k^*,L^*)$ we give the adversary the benefit of an unlimited time budget (implying $(\ref{p-c})=1$) giving
\begin{eqnarray}\label{eq.delta}
u(k^*, L^*)&\geq & \sum_{k=0}^{k^{*}-1} p(k,L^{*}-1,\infty)(\frac{m_4(L^*-1)}{1-m_1(L^*-1)}+\frac{m_3(L^*-1)}{1-m_1(L^*-1)}) \\
&=&  \sum_{B=0}^{\min\{L^*-1,(k^*-2)\}}
\sum_{k=B+1}^{k^{*}-1} (\ref{p-a})\cdot (\ref{p-b})\cdot \frac{m_3(L^*-1)+m_4(L^*-1)}{1-m_1(L^*-1)}.
\end{eqnarray}

We simplify this lower bound by noticing that (\ref{p-a}) does not depend on $k$ and
\begin{eqnarray*}
	\sum_{k=B+1}^{k^{*}-1}(\ref{p-b}) &=& 
	\sum_{k=B+1}^{k^{*}-1}
	{	\sum_{\substack{g_0\geq 1,\dots,g_{L^{*}-1}\geq 1 s.t.\\
				\sum_{l=0}^{L^{*}-1}(g_l-1)=k-1-B}}
		\prod_{l=0}^{L^{*}-1} \bigg(\frac{m_2(l)}{1-m_1(l)} \bigg)^{g_l-1}
	}\\
	&=& \sum_{\substack{g'_0,\dots,g'_{L^*-1}\geq 0 s.t.\\
			\sum_{l=0}^{L^*-1}{g'_l \leq k^*-2-B}}}\prod_{l=0}^{L^*-1} \bigg(\frac{m_2(l)}{1-m_1(l)} \bigg)^{g'_l}\\
	&\geq& \sum_{g'_0=0}^{(k^*-2-B)/L^*} \dots \sum_{g'_{L^*}=0}^{(k^*-2-B)/L^*} 
	\prod_{l=0}^{L^*-1} \bigg(\frac{m_2(l)}{1-m_1(l)} \bigg)^{g'_l} \\
	&=&  \prod_{l=0}^{L^*-1} \sum_{g'_l= 0}^{(k^*-2-B)/L^*}\bigg(\frac{m_2(l)}{1-m_1(l)} \bigg)^{g'_l} \\
	&=&	 \prod_{l=0}^{L^*-1}\frac{1}{1-\frac{m_2(l)}{1-m_1(l)}}\big(1-(\frac{m_2(l)}{1-m_1(l)})^{\frac{k^*-2-B}{L^*}+1} \big) \\
	&=& \prod_{l=0}^{L^*-1} \frac{1-m_1(l)}{m_3(l)+m_4(l)} \bigg(1-(\frac{m_2(l)}{1-m_1(l)})^{\frac{k^*-2-B}{L^*}+1} \bigg)\\
&=&  \bigg( \frac{1-m_1}{m_3+m_4}\bigg)^{L^*} \bigg(1-(\frac{m_2}{1-m_1})^{\frac{k^*-2-B}{L^*}+1} \bigg)^{L^*}.
\end{eqnarray*}
Since $B\leq L^*-1$, this can be further lower bounded as
\begin{eqnarray*}
	\sum_{k=B+1}^{k^{*}-1}(\ref{p-b}) &\geq& 
	\bigg( \frac{1-m_1}{m_3+m_4}\bigg)^{L^*} \bigg(1-(\frac{m_2}{1-m_1})^{\frac{k^*-1}{L^*}} \bigg)^{L^*}.
	\end{eqnarray*}
	
	Notice that this lower bound is independent from $B$, hence,
\begin{eqnarray*}
	u(k^*, L^*)&\geq &	\sum_{B=0}^{\min\{L^*-1,(k^*-2)\}}  (\ref{p-b})\cdot \frac{m_3(L^*-1)+m_4(L^*-1)}{1-m_1(L^*-1)} \\
	&& \hspace{1cm} \cdot  \bigg( \frac{1-m_1}{m_3+m_4}\bigg)^{L^*} \bigg(1-(\frac{m_2}{1-m_1})^{\frac{k^*-1}{L^*}} \bigg)^{L^*} \\
	&=& 	\sum_{B=0}^{\min\{L^*-1,(k^*-2)\}}  (\ref{p-b})\cdot  \\
	&& \hspace{1cm} \cdot  \bigg( \frac{1-m_1}{m_3+m_4}\bigg)^{L^*-1} \bigg(1-(\frac{m_2}{1-m_1})^{\frac{k^*-1}{L^*}} \bigg)^{L^*}.
	\end{eqnarray*}
	
We assume $k^*\geq L^*+1$ and  derive	
\begin{eqnarray*}
	\sum_{B=0}^{\min\{L^*-1,(k^*-2)\}}  (\ref{p-b}) &=&
	\sum_{B=0}^{L^*-1}
	 \sum_{\substack{b_1,\dots,b_{L^*-1}\in\{0,1\} s.t.\\ 
				\sum_{l=1}^{L^*-1}b_l=B}}
		\prod_{l=0}^{L^*-2} \bigg( b_{l+1}\frac{m_3(l)}{1-m_1(l)} + (1-b_{l+1})\frac{m_4(l)}{1-m_1(l)} \bigg) \\
	& =&	 \sum_{b_1,\dots,b_{L^*-1}\in\{0,1\}}\prod_{l=0}^{L^*-2} \bigg( b_{l+1}\frac{m_3(l)}{1-m_1(l)} + (1-b_{l+1})\frac{m_4(l)}{1-m_1(l)} \bigg)
	\\
	&=&\prod_{l=0}^{L^*-2}(\frac{m_3(l)}{1-m_1(l)}+\frac{m_4(l)}{1-m_1(l)})= 
	\bigg(\frac{m_3+m_4}{1-m_1}\bigg)^{L^*-1}.
\end{eqnarray*}
Substituting this in the lower bound for $u(k^*,L^*)$ yields
\begin{eqnarray*}
u(k^*, L^*) &\geq& \bigg(1-(\frac{m_2}{1-m_1})^{\frac{k^*-1}{L^*}} \bigg)^{L^*} \geq 1-L^* (\frac{m_2}{1-m_1})^{\frac{k^*-1}{L^*}} \\
&=& 1- L^* [\frac{(1-\gamma)p}{1-(1-\gamma)(1-(p+h))}]^{\frac{k^*-1}{L^*}}.
\end{eqnarray*}












\subsection{Capacity Region -- Proofs of Corollaries \ref{crcor} and \ref{lemdel}}

\noindent
{\bf Proof of Corollary \ref{crcor}.}
As an example of computing a capacity region, we translate Theorem \ref{Theo.wbarUB} for $1-f(l) \leq \frac{d}{(l+2)^2}$ in terms of a capacity region: We have the security guarantee
\begin{equation}
 T_{bud} \leq e^{-1+k (1-\gamma)/\gamma } \ \ \Rightarrow  \  \ \bar{w}(k,T_{bud}) \leq 
\frac{e^{1/(1-\gamma)}}{1-\gamma} \cdot T_{bud}^{1+1/(1-\gamma)} \cdot  [\frac{1-\gamma}{\gamma}pd]^{k} .
\label{eqimpl}
\end{equation}
If 
\begin{eqnarray}
2^t &\leq & e^{-1+k (1-\gamma)/\gamma }, \mbox{ and} \label{crex}  \\
2^{-s} &\geq& \frac{e^{1/(1-\gamma)}}{1-\gamma} \cdot (2^t)^{1+1/(1-\gamma)} \cdot  [\frac{1-\gamma}{\gamma}pd]^{k} , \nonumber
\end{eqnarray}
then $T_{bud}\leq 2^t$ implies the condition on the left hand side of the implication in (\ref{eqimpl}), hence,
\begin{eqnarray*}
\bar{w}(k,T_{bud}) &\leq & \frac{e^{1/(1-\gamma)}}{1-\gamma} \cdot T_{bud}^{1+1/(1-\gamma)} \cdot  [\frac{1-\gamma}{\gamma}pd]^{k} \\
&\leq & \frac{e^{1/(1-\gamma)}}{1-\gamma} \cdot (2^t)^{1+1/(1-\gamma)} \cdot  [\frac{1-\gamma}{\gamma}pd]^{k} \\
&\leq & 2^{-s}.
\end{eqnarray*}
This shows that the capacity region is characterized by (\ref{crex}), or equivalently after taking logarithms, assuming $\frac{pd(1-\gamma}{\gamma}<1$, and reordering terms, $\boldsymbol{\delta}$, $\boldsymbol{\mu}$, and $\boldsymbol{\xi}$ are given by 2-dimensional vectors
\begin{eqnarray*}
\boldsymbol{\delta} &=& (0, q \ln 2), \mbox{ where } q = \left[ \ln(\frac{\gamma}{pd(1-\gamma)})\right]^{-1}, \\
\boldsymbol{\mu} &=& (\frac{\gamma \ln 2}{1-\gamma}, q (\ln 2)(1+\frac{1}{1-\gamma})), \\
\boldsymbol{\xi} &=& (-\frac{\gamma}{1-\gamma}, q(-\frac{1}{1-\gamma} + \ln \frac{1}{1-\gamma}).
\end{eqnarray*}

\vspace{0.5cm}

\noindent
{\bf Proof of Corollary \ref{lemdel}.}
Delayed learning shows that for all $k^*\geq L^*+1$,
\begin{eqnarray*}
 w(k,T_{bud}) &\leq& (1- u(k^*,L^*)) + w'(k-k^*,T_{bud}) \\
 &\leq&  L^* [\frac{(1-\gamma)p}{1-(1-\gamma)(1-(p+h))}]^{\frac{k^*-1}{L^*}} + w'(k-k^*,T_{bud}).
 \end{eqnarray*}
If 
\begin{eqnarray}
  L^* [\frac{(1-\gamma)p}{1-(1-\gamma)(1-(p+h))}]^{\frac{k^*-1}{L^*}} &\leq& 2^{-(s+1)} \mbox{ and} \label{cond1} \\
 w'(k-k^*,T_{bud}) &\leq& 2^{-(s+1)}, \label{cond2}
\end{eqnarray}
then $w(k,T_{bud})\leq 2^{-s}$. 

We assume $\frac{(1-\gamma)p}{1-(1-\gamma)(1-(p+h))} <1$. Then, after taking logarithms, substituting $k^*=\hat{k}+L^*+1$ with $\hat{k}\geq 0$, and reordering terms,
condition (\ref{cond1}) is equivalent to
\begin{eqnarray*}
(s+1)\delta'' &\leq& \hat{k} +\xi'', \mbox{ where} \\
\delta'' &=& L^* (\ln 2) z \mbox{ and }  
\xi''=L^*[1-(\ln L^*) z] \mbox{ for }
z= \left[ \ln(\frac{1-(1-\gamma)(1-(p+h))}{(1-\gamma)p})\right]^{-1}.
\end{eqnarray*}

Suppose that $w'(\cdot,\cdot)$ corresponds to capacity region $(\boldsymbol{\delta}', \boldsymbol{\mu}',  \boldsymbol{\xi}')$. Then,
condition (\ref{cond2})  is implied by $T_{bud}\leq 2^t$ and
$$ (s+1)\boldsymbol{\delta}' + t \boldsymbol{\mu}' \leq (k-(\hat{k}+L^*+1)) {\bf 1} + \boldsymbol{\xi}'.$$

If we assume learning is delayed sufficiently long such that 
$$ (s+1)\delta'' \geq \xi''$$
(e.g., $L^*\geq z^{-1}/2$ for $s=0$),
then we may choose 
$$\hat{k}=  (s+1)\delta''-\xi''\geq 0.$$
This satisfies condition (\ref{cond1}). Condition (\ref{cond2}) (after substituting $\hat{k}$) is now equivalent to $T_{bud}\leq 2^t$ and
$$ (s+1)(\boldsymbol{\delta}'+\delta'' {\bf 1}) + t \boldsymbol{\mu}' \leq k {\bf 1} + \boldsymbol{\xi}+(\xi''-(L^*+1)){\bf 1}.$$
In other words,
$$ (\boldsymbol{\delta}'+\delta'' {\bf 1}, \boldsymbol{\mu}' , \boldsymbol{\xi}'-\boldsymbol{\delta}' +(\xi''-\delta''-(L^*+1)){\bf 1})$$
is a capacity region.

If we assume learning is moderately delayed such that
$$ (s+1)\delta'' \leq \xi'',$$
then we may choose $\hat{k}=0$ and 
$$ (\boldsymbol{\delta}', \boldsymbol{\mu}' , \boldsymbol{\xi}'-\boldsymbol{\delta}' -(L^*+1){\bf 1}).$$
is a capacity region.

Substituting
$$s\geq \xi''/\delta'' -1 = (L^*[1-(\ln L^*) z])/L^* (\ln 2) z  -1 = (L^*[1-(\ln 2L^*) z])/L^* (\ln 2) z = 1/{z\ln 2} - {\ln (2L^*)}{\ln 2}$$
and
$$\xi''-\delta'' -(L^*+1) = L^*[1-(\ln L^*) z] -  L^* (\ln 2) z -(L^*+1) = L^* (\ln 2L^*) z -1$$
proves the corollary.